\documentclass[10pt, journal]{IEEEtran}
\usepackage{caption}
\usepackage{subcaption}
\usepackage{lipsum}
\usepackage{cite}
\usepackage{array}

\usepackage{color}

\usepackage{graphics}
\usepackage{amsmath,amsfonts, amssymb}
\usepackage{mathtools}
\usepackage{epstopdf}
\usepackage{amsthm}
\usepackage{bbm}
\usepackage{algorithm}
\usepackage[noend]{algpseudocode}
\usepackage{dsfont}
\graphicspath{{/}}

\ifCLASSINFOpdf
\else
\fi

\begin{document}

\title{Throughput Optimal Random Medium Access Control for Relay Networks with Time-Varying Channels}

\author{Mehdi~Salehi~Heydar~Abad,
        Ozgur~Ercetin,
        Eylem~Ekici 
\thanks{ M.S.H.~Abad  and O.~Ercetin are with the Faculty of Engineering and Natural Sciences, Sabanci University, 34956 Istanbul, Turkey.}
\thanks{ E.~Ekici is with the Department of Electronics and Computer Engineering, The Ohio State University, Columbus, OH.}
\thanks{This work is supported in part by a grant from European Commission MC-IRSES programme.}}


\maketitle

\begin{abstract}
The use of existing network devices as relays has a potential to improve the overall network performance. In this work, we consider a two-hop wireless relay setting, where the channels between the source and relay nodes to the destination node are time varying. The relay nodes are able to overhear the transmissions of the source node which may have a weak connection to the destination, and they help the source node by forwarding its messages to the destination on its behalf, whenever this is needed. We develop a distributed scheme for relay selection and channel access that is suitable for time-varying channels, and prove that this scheme is throughput optimal. We obtain the achievable rate region of our proposed scheme analytically for a relay network with a single source and a single relay node.  Meanwhile, for a more general network with more than one relay nodes, we perform Monte-Carlo simulations to obtain the achievable rate region.  In both cases, we demonstrate that the achievable rate region attained with our distributed scheme is the same as the one attained with centralized optimal scheme.
\end{abstract}

\IEEEpeerreviewmaketitle

\section{Introduction}

\subsection{Motivation}
\label{sec:Motivation}


In wireless networks, there is always a node which has the worst channel quality (e.g., cell edge node) due to its location, shadowing effects or physical radio capabilities. Wireless channel quality for these nodes may become so weak that in order to maintain a reliable communication, some
relays have to be employed on the path between the
source and the destination node to aid their communication \cite{RN1}. Instead of deploying separate devices solely used for relaying, it is possible to utilize the existing wireless devices in the network for the same purpose. It is expected that future wireless devices are going to be highly adaptive and flexible, and thus, opening up new paradigms where they can adapt themselves utilizing any network resource, e.g., as other peer nodes in their vicinity, to improve their communications.

One of such paradigms is relay networks where the aim is to improve the network throughput by allowing wireless nodes to participate in the transmission, when they are neither the intended destination nor the source \cite{RN2}. A widely used wireless relay network model consists of two hops, i.e., there are a source, a destination,
and multiple relay nodes \cite{coop-top} as in Figure \ref{sysmod1}. The basic idea is that the relay nodes, overhearing the message transmitted from the source node, forward the same message to the destination node instead of treating it as interference \cite{RN3}.

\begin{figure}
	\centering
		\includegraphics[scale=.5]{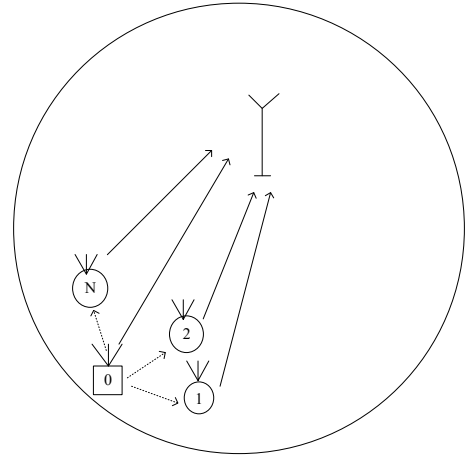}
	\caption{Relay Network.}
	\label{sysmod1}
\end{figure}

Fundamentally, this idea of relaying is very similar to routing in multi-hop networks, and there is already a  rich literature on the subject. In particular, centralized {\em throughput optimal} policies were developed for opportunistic multi-hop relay networks in \cite{tas1,cent-opt,neely-div}. A policy is {\em throughput optimal} in the sense of having an achievable rate region that coincides with the network achievable rate region, and thus, a superset of the achievable rate region of any other policy. However, a centralized policy requires global information of the entire network. The collection of this information is costly in terms of power consumption and time. On the other hand, distributed (decentralized) policies may alleviate these implementation costs by running a loosely coordinated scheme at each network node.

A common network setting is where the source and relay nodes are located in such a way that the source is {\em furthest} to the destination, while the relays are relatively {\em closer}. This notion of distance includes both real distances in free-space communication, or virtual distances that also include the effects of fading and shadowing \cite{ephremiduscog}.  Hence, the source has the worst channel quality in terms of probability of successfully transmitting its packets to the destination.

There is a number of works developing distributed policies for two-hop wireless relay networks with the objective of improving the transmission rate of the source node \cite{pappas,R2,R4,R5}. 
In these works, the time-varying wireless channel is modeled as an ON-OFF random process for analytical tractability, where a packet is successfully decoded at the destination node when the channel is ON or it is lost when the channel is OFF. Note that modeling the channel as an ON-OFF random process is an approximation of a more general time-varying channel model, where the channel is a multi-state process, with each state supporting a maximum transmission rate with which the data can be reliably sent \cite{srikantstable}. The authors then propose \emph{suboptimal} channel access algorithms, and calculate the achievable transmission rate of the source node by employing  Loynes' theorem \cite{loynes}. 

Unlike prior works, we develop a throughput optimal distributed policy called Relay Q-CSMA (RQ-CSMA), for relay networks with time-varying channels, and show that its achievable rate region coincides with the achievable rate region of a centralized optimal policy. Our policy achieves this without the knowledge of the channel statistics, and it performs the scheduling of the source and relay nodes by only utilizing local information without any explicit message exchanges.

\subsection{Related Work}
\label{sec:RelatedWork}

The network setting used in this paper was previously studied under different contexts and assumptions in the literature. In particular, for multi-hop opportunistic routing it is shown that Maximum Weight Scheduling (MWS) algorithm, which is a centralized scheme is throughput optimal \cite{tas1}.  MWS  and many of its variants exist in the literature addressing the optimal resource allocation under different network assumptions and for different applications, e.g., cognitive networks \cite{MWScog}, relay networks \cite{MWSrelay}, etc \cite{MWS1,srikantstable,neelynow}.  Unfortunately, MWS is known to be NP-hard for general networks \cite{np,srofnp} and it is not amenable to a distributed implementation.

The high time complexity of the centralized algorithms such as MWS was addressed in the literature by developing low-complexity sub-optimal algorithms. In particular, maximal scheduling is a low-complexity alternative to
MWS that is amenable to parallel and distributed implementation \cite{alon}. However, maximal scheduling
may only achieve a fraction of the achievable rate region \cite{srikanteff}. Greedy Maximal Scheduling (GMS),
also known as Longest-Queue-First (LQF), is another
low-complexity alternative to MWS \cite{greedgood} with a complexity that grows linearly with the
total number of the links \cite{birand}. Its performance
has been observed to be close to optimum in a variety of
wireless network scenarios \cite{srofimperfect}. Although these algorithms reduce the complexity of MWS, they have two main shortcomings. First, they require network wide queue length information exchange, and second, they are mainly suitable for links with time-invariant properties.

Another class of
distributed scheduling policies, called Queue-length-based Random
Access Scheduling policies, use local message exchanges to
resolve the contention problem \cite{srikantQ,srofQ}. By adjusting each link's
contention probability using the link's local queue information,
it provides explicit tradeoffs between efficiency,
complexity, and the contention period. Carrier-Sensing-Multiple-Access (CSMA)-based scheduling
policies \cite{csma,qcsma}, reduce the complexity by simplifying the comparison process, by
exploiting carrier-sensing. Nonetheless, these
results indicate that good throughput performance may be
attained for time invariant channels using algorithms with
very low complexity.

In practice, however, most wireless systems have time varying characteristics. When link states vary over time, the system throughput can be improved by scheduling those links with better states. This is
known as the opportunistic gain \cite{srofgain}. However, many of the low-complexity scheduling algorithms in the literature do not exploit the opportunistic gain, and their
performance in time varying channels are shown to be significantly lower \cite{joo,srofjoo,yun}.

Recently, a few other
low-complexity schemes have been proposed that are provably efficient over time variable channels \cite{joo,srofjoo}. However, these algorithms, depending on the network structure, require either local or global queue length and channel state information (CSI) exchange. In Opportunistic ALOHA \cite{aloha1,aloha2} the transmission probability is allowed to be a function of the channel state information and maximum throughput of the system is achieved.  Distributed Opportunistic Scheduling (DOS) \cite{dos1,dos2} involves a process of joint channel probing and distributed
scheduling. The authors show that the
optimal scheme for DOS turns out to be a pure threshold policy,
where the rate threshold can be obtained by solving a fixed-point
equation. However both Opportunistic ALOHA and DOS are designed under the assumption of saturated queues and their performance have no guarantees for unsaturated queues. 

The use of relay nodes has also been investigated in the context of cellular networks. In \cite{cell-cent} a throughput optimal centralized downlink scheduling  policy is developed for multi-hop relaying in a cellular network. Note that the uplink scheduling is more challenging, since the queue length information is not readily available at the base station. Conventional cellular architecture treats the mobile users as simple transceivers and due to this rigid client-server scheme, they are fully commanded by cellular base stations. In future cellular architectures, mobile users will be given more freedom, and thus, they will be treated as local micro-operators to improve network coverage \cite{tas-cell}. Device-to-Device (D2D) communication is another paradigm which appears to be a promising component in the next generation cellular networks \cite{d2d-survey}. D2D communication was proposed in \cite{d2d-relay} to enable multi-hop relaying in cellular networks. The potential of D2D communications for improving spectral efficiency of cellular networks was also identified in \cite{d2d-mulhop1,d2d-mulhop2}. There are numerous distributed D2D policies \cite{flash,auction1,auction2,auction3}, addressing the data relay services enabled by existing devices in the network. However, most prior works in the literature ignore the network stability and do not investigate throughput optimality.

\subsection{Main results and organization}
\label{sec:MainResultsAndOrganization}
In this work, we take both channel state and queue length information into account when scheduling the source and relay nodes. We propose a distributed CSMA based scheduling algorithm which is analytically shown to be the throughput optimal. The time complexity of the algorithm is shown to be linear in the number of nodes. 
The rest of the paper is organized as follows. In Section \ref{sec:SystemModel}, we introduce the relay network model. In Section \ref{sec:DistAlg}, we give a brief introduction to MWS and later propose our scheduling and contention resolution algorithms. In Section \ref{sec:ThroughputOptimality}, we analytically prove throughput optimality of our algorithm. We establish the achievable rate region for a source and a single relay node  in Section \ref{sec:CapacityRegionOfTheCooperativeNetwork} and finally, we evaluate the performance of
different scheduling algorithms via simulations in Section \ref{sec:Validnum}.

\section{System Model}
\label{sec:SystemModel}
In this section, we introduce the details of our system model and the definitions of the main concepts that are used in relation to this model. Throughout the paper, the following notations
are adopted. Upper case and lower case bold symbols denote
matrices and vectors, respectively. Upper case calligraphic symbols denote sets. 
 
\subsection{Network Model}
\begin{figure}
	\centering
		\includegraphics[scale=.5]{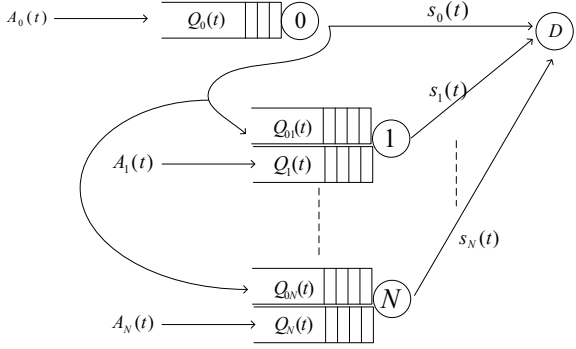}
	\caption{Relay Network. Channels between nodes and destination are ON-OFF whereas the channels between all pairs of nodes are always ON.}
	\label{sysmod}
\end{figure}
We consider a wireless relay network consisting of a source node and $N$ relay nodes all wishing to communicate with a common destination node as shown in Figure \ref{sysmod}. We consider a time-slotted system where the time slot is the resource to be shared among different nodes. We adopt a non-interference model where only one node is transmitting at a given time\footnote{This assumption is widely used in the literature, e.g., it is proposed that the cellular spectrum is utilized for both D2D and cellular communications (i.e., underlay inband D2D) \cite{d2d-survey}.}. The length of the time slot is equal to the transmission time of a single packet.

We assume that the channels between $N+1$ nodes and the destination node are time-varying ON-OFF channels, where a packet is either successfully decoded at the destination or lost with a certain probability. 
The nodes with better channel conditions to the destination have higher successful transmission probabilities. We will refer to the source node as node $0$ and $N$ relay nodes as node $i$, $i = 1,\ldots,N$. We assume that relay nodes are in close proximity to node $0$ so that the channels between node $0$ and nodes $i=1,\ldots,N$ are always ON and node $0$ can always transmit successfully to any node $i=1,\ldots,N$. This assumption is realistic in practical scenarios (e.g., see \cite{lterelay}), due to proximity of source node and relays.


Let $s_{i}(t)\in\{0,1\}$ be the channel state between node $i$, $i=0,1,\ldots,N$ and the destination at time slot $t$.  The random channel states, $s_{i}(t)$, are assumed to be independent and identically distributed (iid) across time and nodes. Let $\rho_{i}$ be the probability that the channel between node $i$ and the destination is ON, i.e., $s_{i}(t)=1$\footnote{$\rho_i= \mathsf{Pr}\left[\frac{\left|h_i\right|^2P}{N_0}>\gamma\right]$, where $h_i$ is the channel gain from node $i$ to the destination node, $N_0$ is the
noise power level at the destination and $\gamma$ is the Signal-to-Noise
Ratio (SNR) threshold required for correct decoding.}.

Let $A_{i}(t)$ be the number of packets arriving to node $i=0,1,\ldots,N$ at time $t$. The arrival processes are assumed to be iid across time and nodes. Let $\lambda_{i}=\mathsf{E}\left[A_{i}(t)\right]$ be the rate of arrival to node $i$. We also assume that $A_{i}(t)\leq A_{max}$ and $\mathsf{E}\left[A_{i}(t)^{2}\right]$ is finite for all $i$ and $t$. The incoming packets are stored in a queue until they are finally transmitted.

Let $Q_{i}(t)$ be the size of the queue storing the packets of node $i=0,1,\ldots,N$ at time $t$. Meanwhile, each node $i =1,\ldots,N$ also keeps a separate queue for node $0$ packets that they are supposed to relay on its behalf.  Let $Q_{0i}(t)$, $i=1,2,\ldots,N$ be the size of this queue at time $t$.

Finally, we assume that the transmission rate of each node is the same and equal to one unit sized packet per time slot.

\subsection{The Relaying Scheme}
At the beginning of time slot $t$, the channel is acquired by one of the nodes according to the scheme to be explained in the subsequent section. If node $0$ acquires the channel, it sends its packet to the destination if $s_0(t)=1$, i.e., the channel between node $0$ and destination is ON.  Otherwise, it forwards its packet to one of the nodes $1,\ldots,N$ for future delivery of the packet to the destination via the selected node $i$, $i=1,2,\ldots,N$. Node $0$ forwards the packet to node $i^{*}$, where  $i^{*} = \operatorname{arg\,max}_j \left\{Q_{0}(t)-Q_{0j}(t)\right\}$.
Meanwhile, if a node other than node $0$ acquires the channel in a time slot, it may transmit either one of its own packets or a packet received from node $0$ in previous time slots. Node $i$, $i=1,\ldots,N$ transmits a packet from one of its two queues, i.e., $Q_{0i}(t)$ or $Q_i(t)$, with the largest backlog. The destination sends an Acknowledgment (ACK) message after the reception of the packet.  ACK message is perfectly received by both node $0$ and the transmitting node $i$, $i=1,\ldots,N$, and thus, node $0$ is cognizant of the delivery of its packet by the node $i$, $i=1,2,\ldots,N$. Note that the assumption of perfect reception of ACK/NACK is not necessary. It can be shown in a similar way as was done in \cite{srikantstable} that as long as the queues are updated periodically with period $T>1$, our results still hold.

Let $x(t)=\mathfrak{I},0,1,\ldots,N$ be the scalar denoting the scheduled node for transmission, where $\mathfrak{I}$ denotes the \emph{Idle} state indicating no nodes are scheduled for transmission. Based on the aforementioned relaying scheme, the queues evolve as follows:
\begin{align}
Q_{0}(t+1)=&\max\left[Q_{0}(t)-\mathbbm{1}_{\left\{x(t)=0\right\}},0\right]+A_{0}(t),\label{Q0}\\ 
Q_{i}(t+1)=&\max\left[Q_{i}(t)-\mathbbm{1}_{\left\{x(t)=i\right\}}\cdot\mathbbm{1}_{\left\{Q_{i}(t)\geq Q_{0i}(t)\right\}},0\right]\nonumber\\
&+A_{i}(t),\hspace{2cm}\forall i=1,\ldots, N\label{Qi}\\
Q_{0i^*}(t+1)=&\max\left[Q_{0i^*}(t)-\mathbbm{1}_{\left\{x(t)=i^*\right\}}\cdot\mathbbm{1}_{\left\{Q_{i^*}(t)< Q_{0i^*}(t)\right\}},0\right]\nonumber\\
&+\left(1-s_{0}(t)\right)\mathbbm{1}_{\left\{x(t)=0\right\}}, \nonumber\\
&\hspace{0.2cm}\mbox{ for } i^* = \operatorname{arg\,max}_j \left\{Q_{0}(t)-Q_{0j}(t)\right\},\label{Q0i}\\
Q_{0i}(t+1)=&\max\left[Q_{i}(t)-\mathbbm{1}_{\left\{x(t)=i\right\}}\cdot\mathbbm{1}_{\left\{Q_{i}(t)< Q_{0i}(t)\right\}},0\right]\nonumber\\
&\forall i\neq i^*,
\end{align}
where $\mathds{1}_{\{\cdot\}}$ is the indicator function, and its value is equal to $1$ if the condition inside the bracket is true, and $0$ otherwise. 


%

\section{Distributed Algorithm}
\label{sec:DistAlg}
\subsection{Background and Preliminaries}
\label{sec:BackgroundAndPreliminaries}

Before we proceed with the discussion of our proposed distributed scheduling algorithm, we give a brief overview of the well-known Maximum Weight Scheduling (MWS) \cite{neelynow} and Q-CSMA algorithms \cite{qcsma} in the context of our relay network. 
\subsubsection*{Maximum Weight Scheduling}
According to MWS algorithm, each node is assigned a weight that is the product of its {\em differential} queue length and the channel state. MWS algorithm chooses the node with the maximum weight for transmission at time slot $t$.  The differential queue length of the node is the difference of the queue lengths of the ingress and egress nodes. Note that in our relay network, the destination is the sink node, and thus, its queue length is always zero.  Let $Q_{r}(t)$ be the maximum differential backlog of node $0$ with respect to nodes $1,\ldots,N$ at time slot $t$, which is defined as
\begin{align}
Q_{r}(t) = \max_{i=1,2,\ldots,N} \left\{Q_{0}(t)-Q_{0i}(t),0\right\}.
\end{align}
In our network model, we assign weights to nodes in the network, and aim to schedule the node (and the queue of that node) with the highest weight at every slot.  The weight of node $i$ at slot $t$, $\omega_{i}(t)$, is defined as follows:
\begin{align}
&\omega_{0}(t) = f_{0}(Q_{0}(t))s_{0}(t)+f_{0}(Q_{r}(t))(1-s_{0}(t)), \label{W0}\\ 
&\omega_{i}(t) = \max\left\{f_{i}(Q_{0i}(t)),f_{i}(Q_{i}(t))\right\}s_{i}(t)\ ,\forall i\neq 0,\label{Wi}
\end{align}
where $f_{i}:[0,\ \infty]\rightarrow[0,\ \infty]$, $i=0,1,\ldots,N$ are functions that should satisfy the following conditions \cite{srikantstable}:
\begin{enumerate}
	\item $f_{i}(Q_{i})$ is a non decreasing, continuous function with $\lim_{Q_{i}\rightarrow\infty}f_{i}(Q_{i})=\infty$.
	\item Given any $M_{1}>0$, $M_{2}>0$ and $0<\epsilon<1$, there exists a $Q<\infty$, such that for all $Q_{i}>Q$ and $\forall i$, we have
	\begin{align}
	(1-\epsilon)f_{i}(Q_{i})\leq f_{i}(Q_{i}-M_{1})&\leq f_{i}(Q_{i}+M_{2})\nonumber\\
	&\leq (1+\epsilon)f_{i}(Q_{i}).
	\end{align}
\end{enumerate}

It can be shown that for this system model, the centralized MWS algorithm is throughput optimal using the Lyapunov drift theorem \cite{neelynow}. However, note that MWS requires the complete instantaneous queue length and channel state information.

\subsubsection*{Queue based Carrier-Sensing-Multiple-Access (Q-CSMA)}
Carrier-Sensing-Multiple-Access (CSMA)-based scheduling policies \cite{csma,qcsma}, reduce the complexity of MWS by simplifying the comparison of the weights exploiting the carrier-sensing mechanism.  In Q-CSMA, each time slot $t$ is divided into a \emph{contention} and a \emph{data} 
slot. For any node $i$, let $\mathcal{C}(i)$ be the set of
conflicting nodes (called conflict set) of node $i$, i.e., $\mathcal{C}(i)$ is the set of
nodes such that if any one of them is active, then node $i$ cannot be active. In the contention slot, the network first selects a set of nodes that do not conflict with each other,
denoted by $\mathcal{D}(t)$. Note that this set of nodes is also a feasible
schedule, but this is not the schedule to be used for data transmission. $\mathcal{D}(t)$ is called the \emph{decision schedule} at time $t$.

The network selects a decision schedule according to a randomized procedure, i.e., it selects $\mathcal{D}(t)\in \mathcal{M}$
with  probability $\alpha(\mathcal{D}(t))$, where $\sum_{\mathcal{D}(t)\in \mathcal{M}}\alpha(\mathcal{D}(t))=1$, and $\mathcal{M}$ is the set of all possible feasible schedules of the network. Then, the scheduling procedure proceeds as follows: Each node within
$\mathcal{D}(t)$ is checked to decide whether it will be
included in the transmission schedule $x(t)$. For any $i\in\mathcal{D}(t)$ if no nodes in $\mathcal{C}(i)$ were active in the previous data slot, i.e., $\forall j \in \mathcal{C}(i), x(t-1)\neq j$, then $i$ chooses to be \emph{active} with probability $p_{i}$ and inactive with probability $1-p_{i}$ in the current data slot. If at least one node in $\mathcal{C}(i)$
was active in the previous data slot, i.e., $\exists j\in \mathcal{C}(i), x(t-1)= j$, then $i$ will be inactive
in the current data slot. Any node $i\notin\mathcal{D}(t)$
maintains its state (active or inactive) from the previous
data slot.

\subsection{Relay Q-CSMA Algorithm (RQ-CSMA)}
\label{sec:AlgorithmDescription}

We propose RQ-CSMA algorithm to obtain a feasible throughput optimal schedule, when the channels to the destination are randomly varying over time.
In RQ-CSMA, a time slot is divided into a contention and a transmission slot as in Q-CSMA. Unlike Q-CSMA, the purpose of the contention slot in RQ-CSMA is not only to generate a collision free schedule to be used in the transmission slot but also to infer the channel states of every node, i.e., $\textbf{s}(t)=\left(s_{0}(t),\ s_{1}(t),\ldots,\ s_{N}(t)\right)$. 


Let us denote by $\mathcal{D}^{\boldsymbol{\vartheta}}(t)$, the decision schedule at time slot $t$ with $\textbf{s}(t)=\boldsymbol{\vartheta}$, where $\boldsymbol{\vartheta}=\left\{\vartheta_{0},\ \vartheta_{1},\ldots,\ \vartheta_{N}\right\}$, and $\vartheta_{i}\in \{0,1\}$ is a realization of randomly varying channel states.  Note that in the relay network considered, the cardinality of set $\mathcal{D}^{\boldsymbol{\vartheta}}(t)$ is always one, since only one node can be scheduled without conflicting with others.  The network selects a decision schedule $\mathcal{D}^{\boldsymbol{\vartheta}}(t)$ randomly with probability $\alpha\left(\mathcal{D}^{\boldsymbol{\vartheta}}(t)\right)$, where $\sum_{\mathcal{D}^{\boldsymbol{\vartheta}}(t)\in\mathcal{M}^{\boldsymbol{\vartheta}}}\alpha\left(\mathcal{D}^{\boldsymbol{\vartheta}}(t)\right)=1$, and $\mathcal{M}^{\boldsymbol{\vartheta}}$ is the set of all possible feasible schedules of the network when $\textbf{s}(t)=\boldsymbol{\vartheta}$. Note that in the relay network considered, a schedule is feasible, if only one node is scheduled, and the scheduled node $i=0,1,\ldots, N$ has its channel in ON state, i.e., $s_i(t)=1$. Meanwhile, it is always feasible to schedule node $0$, since links between node $0$ and other nodes are always ON.

Unlike Q-CSMA, in RQ-CSMA, the transmission schedule is determined according to the 
channel states. Let $y^{\mathbf{\mathbf{\vartheta}}}(t)$ be the transmission schedule in the most recent data slot (before and including time $t$) when 
$\mathbf{s}(t)=\mathbf{\vartheta}$. More specifically,
\begin{align}
y^{\boldsymbol{\vartheta}}(t)=\left\{x(\tau^{*}):\ \ \tau^{*}=\max \left\{\tau\ \text{s.t.}\ \tau\leq t\ \ \text{and}\ \ \mathbf{s}(\tau)=\boldsymbol{\vartheta}\right\}\right\}.\label{yS}
\end{align}

$y^{\boldsymbol{\vartheta}}$ is the transmission activity of the nodes in the last instance of time, before and including $t$, where the channel state vector realization was $\vartheta$. Consider the following example to better illustrate the evolution of $y^{\boldsymbol{\vartheta}}$. Assume that only the channel of node $0$ is time varying. Since there is only one time varying channel, there are only two variables defined by \eqref{yS}, i.e., $y^1(t)$ and $y^0(t)$. The sample path evolution of $y^1(t)$, $y^0(t)$ and $x(t)$ are given in Fig. \ref{eq9fig}. Suppose that at time $t=8$, RQ-CSMA has chosen to schedule an \textit{idle} slot. Note that since at $t=8$, the channel state is $0$, $y^0(8)$ updates it value to the value of $x(8)$ while $y^{1}(8)$ maintains its value from $t=7$, since it is not allowed to change its value. Basically, when $s_{0}(t)=0$, only $y^0(t)$ evolves, and $y^1(t)$ does not evolve. Note that in scheduling the nodes, if $s_0(t)=1$ ($s_0(t)=0$), RQ-CSMA uses $y^1(t)$ ($y^0(t)$).

%
At time slot $t$, based on the realization of the channel state vector $\textbf{s}(t)=\boldsymbol{\vartheta}$ a decision schedule, $\mathcal{D}^{\boldsymbol{\vartheta}}(t)$ is randomly determined. Then each node $i$ searches for the corresponding $y$ to the channel state vector $\textbf{s}(t)=\boldsymbol{\vartheta}$, i.e., $y^{\boldsymbol{\vartheta}}$. If node $i\in \mathcal{D}^{\boldsymbol{\vartheta}}(t)$, and $y^{\boldsymbol{\vartheta}}(t-1) \neq j$ for all $j\neq i$, then node $i$ chooses to be active with probability $p^{\boldsymbol{\vartheta}}_{i}$, and idle with probability $1-p^{\boldsymbol{\vartheta}}_{i}$. If $y^{\boldsymbol{\vartheta}}(t-1) = j$ for any $j\neq i$, then $i$ will be inactive in the current transmission slot. Any node $i$ not selected in $\mathcal{D}^{\boldsymbol{\vartheta}}(t)$ will maintain its previous state, i.e., it will be inactive if $y^{\boldsymbol{\vartheta}}(t-1)\neq i$ and active otherwise. The conditions on activation probabilities $p^{\boldsymbol{\vartheta}}_{i}$ will be specified in Section \ref{PTO}. Note that in Q-CSMA, transmission schedule is determined based on the previous transmission schedule $x(t-1)$, whereas in RQ-CSMA we use $y^{\boldsymbol{\vartheta}}(t-1)$ for each $\boldsymbol{\vartheta}$ to determine the transmission schedule $x(t)$. In Section \ref{PTO}, We will show that defining $y^{\boldsymbol{\vartheta}}$ as such, and using it to schedule the transmissions, results in Markov chains with product-form stationary distributions. This is then used to prove throughput optimality of RQ-CSMA. The detailed description of RQ-CSMA is given in Algorithm \ref{alg}.

\begin{algorithm}
\caption{RQ-CSMA}\label{alg}
\begin{text}
At each time slot, each node $i$ performs the following procedure.
\end{text}
\begin{algorithmic}[1] 
\State In \textit{control} slot, randomly select a decision schedule $\mathcal{D}^{\boldsymbol{\vartheta}}(t)\in \mathcal{M}^{\boldsymbol{\vartheta}}$ with probability $\alpha\left(\mathcal{D}^{\boldsymbol{\vartheta}}(t)\right)$ (node $i$ contends to be included in the decision schedule.),
\If{$i\in \mathcal{D}^{\boldsymbol{\vartheta}}(t)$ and $y^{\boldsymbol{\vartheta}}(t-1)\neq j$ for all $j\neq i$}
\State $x(t) = i$ with probability $p^{\boldsymbol{\vartheta}}_{i}$
\State $x(t) = \mathfrak{I}$ with probability $1-p^{\boldsymbol{\vartheta}}_{i}$
\ElsIf{$i\in \mathcal{D}^{\boldsymbol{\vartheta}}(t)$ and $y^{\boldsymbol{\vartheta}}(t-1)=j$ for some $j\neq i$}
\State node $i$ does not transmit.
\Else
\State node $i$ transmits if $y^{\boldsymbol{\vartheta}}(t-1)=i$, otherwise node $i$ does not transmit\ \   
\EndIf
\State In the \textit{data} slot, use $x(t)$ as the transmission schedule
\end{algorithmic}
\end{algorithm}

\begin{figure}
        \centering
        \begin{subfigure}[t]{0.3\textwidth}
                \includegraphics[width=\textwidth]{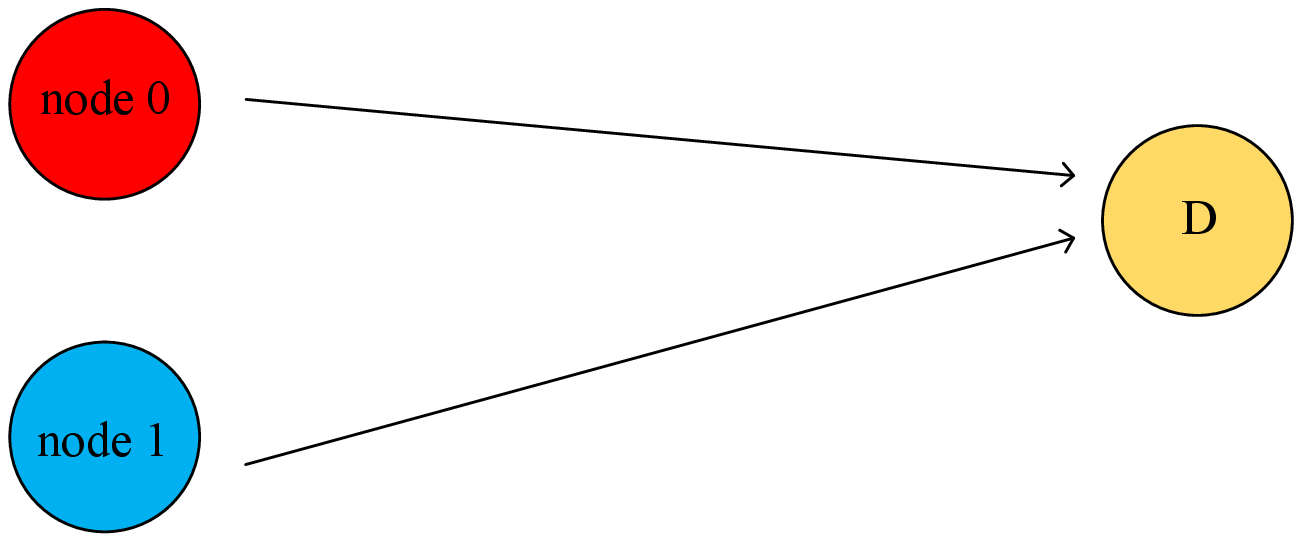}
                \caption{Network topology when $s_{0}(t)=1$}
								\label{systop1}
        \end{subfigure}%
        \hfill
        \begin{subfigure}[t]{0.3\textwidth}
                \includegraphics[width=\textwidth]{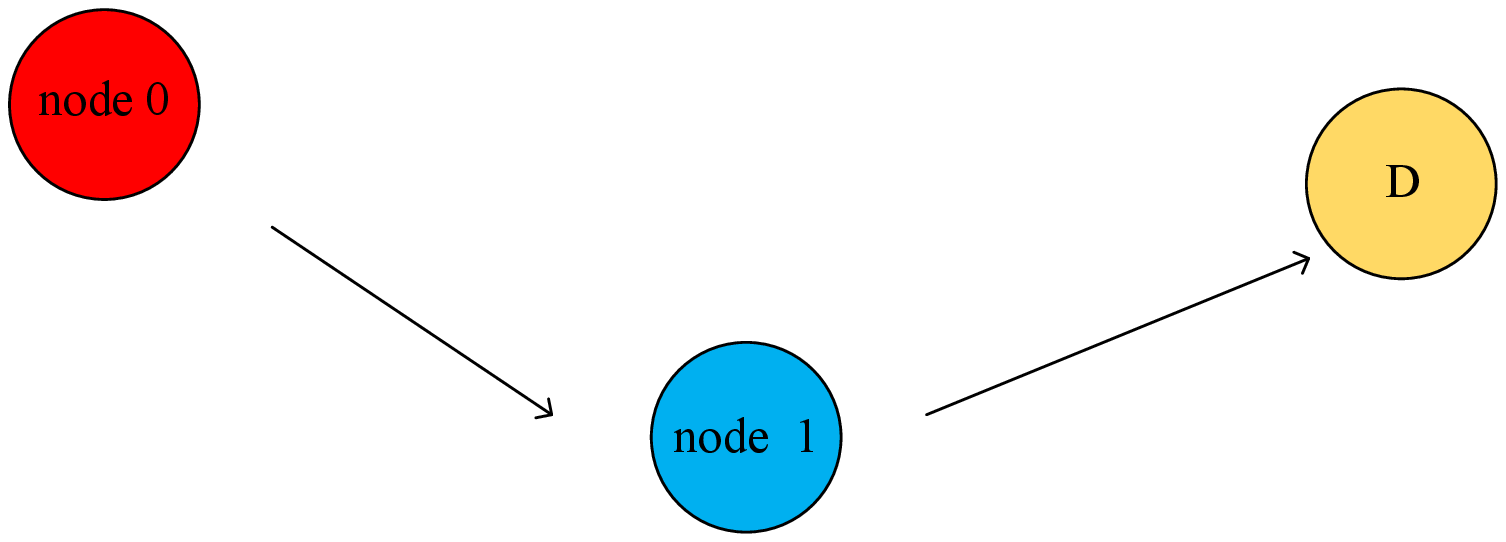}
                \caption{Network topology when $s_{0}(t)=0$}
                \label{systop0}
        \end{subfigure}
				\caption{Different network topologies associated with $s_{0}(t)$ }\label{systop}
\end{figure}

\begin{figure*}
	\centering
		\includegraphics[scale=.8]{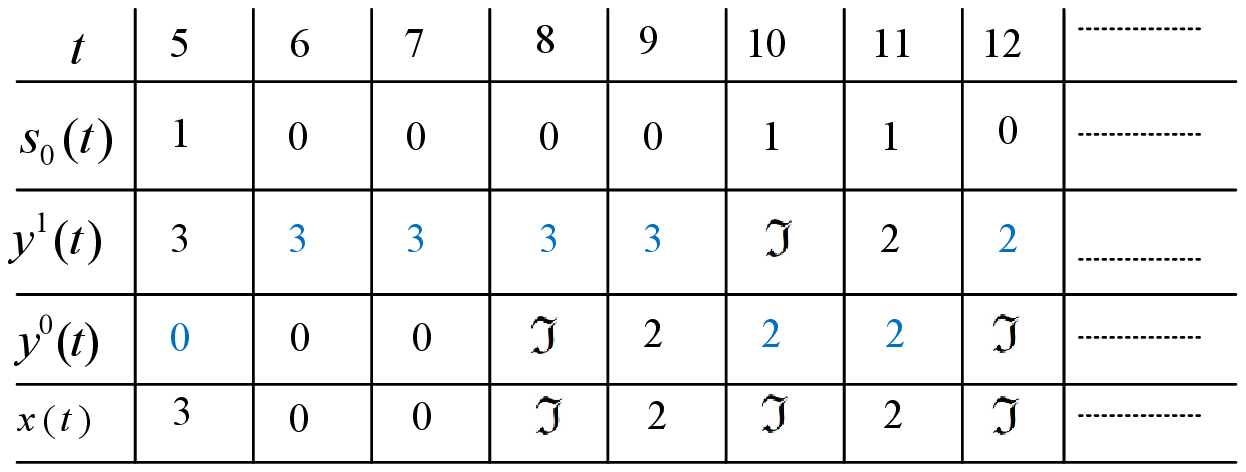}
	\caption{The evolution of $y^1(t)$, $y^0(t)$ and $x(t)$. Numbers in black mean that the schedule is allowed to evolve, while the numbers in blue mean that it remains the same.}
	\label{eq9fig}
\end{figure*}

We illustrate the operation of RQ-CSMA with the following simple example. 
\subsubsection*{Example 1}

Suppose that the network consists of a source and a relay node with only node $0$ having a time varying channel, i.e., $\textbf{s}(t)=s_{0}(t)$. Note that whenever $s_{0}(t)=1$, the network consists of two nodes willing to transmit to a common destination as in Figure \ref{systop1}. In this scenario, the possible decision schedules are $\mathcal{D}^{1}=\left\{\mathfrak{I}\right\}\ \text{or}\ \left\{0\right\}\ \text{or}\ \left\{1\right\}$, which translates into two nodes being idle, or direct transmission of node $0$ to destination or transmission of node $1$ to destination, respectively. Meanwhile, if $s_{0}(t)=0$, the network consists of a node $0$ which forwards its packet to the node $1$ for eventual delivery to the destination as shown in Figure \ref{systop0}. In this scenario, the possible decision schedules are $\mathcal{D}^{0}=\left\{\mathfrak{I}\right\}\ \text{or}\ \left\{0\right\}\ \text{or}\ \left\{1\right\}$, which translates into two nodes being idle, or node $0$ relaying a packet to node $1$ or transmission of node $1$ to destination, respectively. 

Since there is only one time varying channel, we need to define $y^{1}(t)$, the transmission schedule corresponding to the last instance $\tau^*$ that $s_0(\tau^*)=1$ and $y^{0}(t)$, the transmission schedule corresponding to the last instance $\tau^*$ such that $s_0(\tau^*)=0$ according to (\ref{yS}). Note that $y^{1}(t)$ is associated with Figure \ref{systop1} and $y^{0}(t)$ is associated with Figure \ref{systop0}.
 
First consider the case when $s_{0}(t)=1$, $y^{1}(t-1)=\mathfrak{I}$, i.e., in the most recent time slot, $\tau$, when $s_{0}(\tau)=1$, where $\tau<t$, both nodes did not transmit. Then,
	 \begin{itemize}

		 \item if the decision schedule in the current slot is $\mathcal{D}^{1}(t)=\left\{0\right\}$, then node $0$ transmits to the destination with probability $p^{1}_{0}$ and remains idle with probability $1-p^{1}_{0}$.
		\item if the decision schedule is $\mathcal{D}^{1}(t)=\left\{1\right\}$, then node $1$ transmits to the destination with probability $p^{1}_{1}$, and it remains idle with probability $1-p^{1}_{1}$.
	 \end{itemize}
	
Now, consider the other possibility for the channel state of node $0$, i.e., when $s_{0}(t)=0$, we have $y^{0}(t-1)=1$, which means that in the most recent time slot, $\tau$, when $s_{0}(\tau)=0$, node $0$ was idle whereas node $1$ transmitted. Then,

		 \begin{itemize}
		
		 \item if the decision schedule in the current slot is $\mathcal{D}^{0}(t)=\left\{0\right\}$, then node $0$ does not transmit, because the last time when the node $0$ to destination channel was OFF, node $1$ was active.
		\item if the decision schedule in the current slot is $\mathcal{D}^{0}(t)=\left\{1\right\}$, then node $1$ transmits to the destination with probability $p^{0}_{1}$, and it remains idle with probability $1-p^{0}_{1}$.
	 \end{itemize}

	
%

\subsection{Contention Resolution Scheme}
\label{sec:Contention}


In this section, we propose a contention resolution scheme which satisfies two purposes: i) Randomly choosing a feasible schedule, and ii) Inferring the channel state vector $\textbf{s}(t)$ at each node in a \emph{distributed manner}. At the beginning of contention slot, the destination node broadcasts a probe packet.  Upon successful detection of this packet, node $i$ determines that its channel to the destination node is ON at slot $t$, and otherwise, it recognizes that its channel is OFF. We assume that the contention slot is divided into  $N+1+W$ number of mini slots. Each node is registered to one of the first $N+1$ mini slots, i.e., we assign $i^{th}$, $i=0,\ldots, N+1$, mini slot to $i^{th}$ node. At $i^{th}$ mini slot, node $i$ broadcasts a random signal, and all other nodes attempt to detect the random message. If they detect a message on $i^{th}$ mini-slot, they conclude that $s_i(t)=1$, otherwise $s_i(t)=0$.  Hence, at the end of $N+1$ mini-slots, all nodes are aware of other nodes' channel states to the destination.  

To randomly choose (i.e., with probability $\alpha\left(\mathcal{D}^{\boldsymbol{\vartheta}}(t)\right)$ for any $\boldsymbol{\vartheta}\in\mathcal{S}$) a decision schedule $\mathcal{D}^{\boldsymbol{\vartheta}}(t)$, node $i$ (given that $s_i(t)=1$) randomly selects a number $T_{i}$ uniformly distributed in $\left[N+1,N+1+W\right]$, and waits for $T_{i}$ mini slots. If node $i$ hears an INTENT message before $T_{i}+1$, it will not be included in $\mathcal{D}^{\boldsymbol{\vartheta}}(t)$.  Otherwise, node $i$ broadcasts an INTENT message at the beginning of the $(T_{i}+1)^{th}$ control mini-slot. If there is no collision in $(T_{i}+1)^{th}$ control mini-slot, node $i$ will be included in $\mathcal{D}^{\boldsymbol{\vartheta}}(t)$.  Finally, if there is a collision, none of the nodes will be included in $\mathcal{D}^{\boldsymbol{\vartheta}}(t)$ (i.e., $\mathcal{D}^{\boldsymbol{\vartheta}}(t)=\left\{\mathfrak{I}\right\}$).

Note that the contention resolution scheme requires all nodes to hear each other which may not be always feasible. This is similar to well-known {\em hidden node} problem in wireless networks which can be resolved by Request-To-Decide (RTD) and Clear-To-Decide (CTD) mechanism.
The INTENT message in our original scheme is split into RTD and CTD pair transmitted and received, respectively, in two successive sub-mini slots. In the first sub-mini slot, node $i$, $i=1,\ldots,N$ sends RTD to the destination node. If the destination receives RTD without a collision (i.e., no other nodes are transmitting in the same sub-mini slot), then the destination replies with CTD to node $i$ in the second sub-mini slot. If node $i$ receives the CTD from destination without a collision, then node $i$ is added to the decision schedule. 



\section{Throughput Optimality of RQ-CSMA}
\label{sec:ThroughputOptimality}

\begin{figure}
	\centering
		\includegraphics[scale=.13]{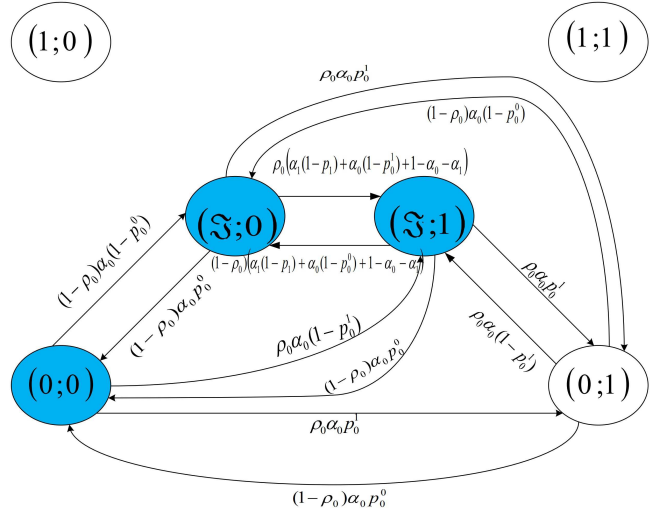}
	\caption{DTMC with ($x;s_{0}$) as states for {\em Example 1}.}
	\label{DTMCex2}
\end{figure}

In this section, we first discuss why a direct application of Q-CSMA is not possible for our relay network.  Then, we explain how we can modify the standard analysis of throughput optimality according to the structure of our proposed RQ-CSMA algorithm.  Finally, we provide the proof of throughput optimality of RQ-CSMA when the channel states are determined correctly according to the contention algorithm discussed earlier. The key notations used in this paper are listed in Table \ref{notation}.

\begin{table}[]
\centering
\caption{SUMMARY OF NOTATION}
\label{notation}
\begin{tabular}{|l|l|}
\hline
Notation                                       & Description                                                                                                                                                            \\ \hline
$s_{i}(t)$                                     & Channel state between node $i$ and the destination at time $t$,                                                                                                        \\ \hline
$A_{i}(t)$                                     & Number of packets arriving to node $i$ at time $t$,                                                                                                                    \\ \hline
$Q_{i}(t)$                                     & Size of the queue, storing the packets of node $i$ at time $t$,                                                                                                         \\ \hline
$Q_{0i}(t)$                                    & Relaying queue of node $i$ at time $t$,                                                                                                                                \\ \hline
$x(t)$                                         & Scheduled node at time $t$,                                                                                                                                             \\ \hline
$\mathcal{D}(t)$                               & The decision schedule,                                                                                                                                                 \\ \hline
$\mathcal{D}^{\mathbf{\mathbf{\vartheta}}}(t)$ & The decision schedule when $\mathbf{s}(t)=\mathbf{\vartheta}$,                                                                                                         \\ \hline
$y^{\mathbf{\mathbf{\vartheta}}}(t)$           & \begin{tabular}[c]{@{}l@{}}Transmission schedule in the most recent data slot \\ (before and including time $t$) when $\mathbf{s}(t)=\mathbf{\vartheta}$,\end{tabular} \\ \hline
$p_i^{\mathbf{\mathbf{\vartheta}}}(t)$         & Activation probability of node $i$ when $\mathbf{s}(t)=\mathbf{\vartheta}$,                                                                                          \\ \hline
$\omega_x(t)$                                  & Weight associated with scheduled node $x$ at time $t$.                                                                                                                \\ \hline
\end{tabular}
\end{table}
\subsection{Drawback of Q-CSMA}
Even though Q-CSMA is known to achieve a throughput-optimal schedule in a distributed fashion for a wireless network with non-time-varying links, it is not a suitable choice for many other network models. A particular case was studied in \cite{qcsmaeylem}, where the authors investigated a random access scheme for scheduling nodes in a CRN.  The authors in \cite{qcsmaeylem} demonstrated that Q-CSMA is not throughput optimal in this network model, since the Discrete Time Markov Chain (DTMC) with the transmission schedule $x(t)$ chosen as the state is not time-reversible, and thus, it does not have a product-form solution.  

Similarly, Q-CSMA is not a suitable choice for our system model as well. Note that Q-CSMA is throughput optimal when the activation probabilities for each node are chosen to be a function of the node weights as defined in (\ref{W0}) and (\ref{Wi}) of \cite{qcsma}. It can be argued that the stationary distribution of DTMC exists in steady state, if the activation probabilities change slowly over time. Unfortunately, in our system model, the node weights may vary abruptly  depending on the channel state. Specifically, in {\em Example 1}, whenever $s_{0}(t)=1$, the weight of node $0$ is $f_{0}(Q_{0}(t))$ and it becomes $f_{0}(Q_{0}(t)-Q_{01}(t))$ otherwise. Consider the case where $Q_{0}(t)$ and $Q_{0i}(t)$ are both very large and equal. In this case, the activation probability of node $0$ is approximately equal to $1$ when $s_{0}(t)=1$ and $0$ otherwise. This results in a transition probability jumping from $1$ to $0$ from one slot to the next. This implies that the distribution of DTMC cannot be assumed to be at steady state at a given time\footnote{Henceforth, Proposition 2 in \cite{qcsma} which assumes that the DTMC is in steady state at any given time does not hold for our network model.}. Moreover, the DTMC in our system model is not time-reversible either. 

\subsection{Modifying the DTMC with channel state information}
Next, we demonstrate that associating the channel state of the node $0$ into the states cannot resolve this problem. 
For this purpose, we again consider the network in {\em Example 1}.  Let $(x;s_{0})$ be the states of the new DTMC as shown in Figure \ref{DTMCex2}\footnote{For the sake of simplicity, we only show the transition probabilities of the node $0$.}.  Consider the highlighted states in Figure \ref{DTMCex2}, and let us check whether the DTMC is time reversible by
examining the transitions from $(0;0)$, $(\mathfrak{I};0)$ to $(\mathfrak{I};1)$ in the clockwise and counter-clockwise directions. Let $\alpha_i$, $i=0,1$, be the probability that node $i$ is included in the decision schedule. The product
of clockwise transition probabilities is $\rho_{0}(1-\rho_{0})^{2}\alpha_{0}^{2}p^{0}_{0}(1-p^{0}_{0})\left(\alpha_{1}(1-p_{1})+\alpha_{0}(1-p^{1}_{0})+1-\alpha_{0}-\alpha_{1}\right)$ and the product of counter-clockwise transition probabilities is $\rho_{0}(1-\rho_{0})^{2}\alpha_{0}^{2}p^{0}_{0}(1-p^{0}_{0})\left(\alpha_{1}(1-p_{1})+\alpha_{0}(1-p^{0}_{0})+1-\alpha_{0}-\alpha_{1}\right)$. Since, $p^{1}_{0}\neq p^{0}_{0}$, these two probabilities are not equal and therefore the DTMC is not time reversible by Kolmogorov$'$s criterion \cite{Kelly}.


\begin{figure}
        \centering
        \begin{subfigure}[t]{0.3\textwidth}
                \includegraphics[width=\textwidth]{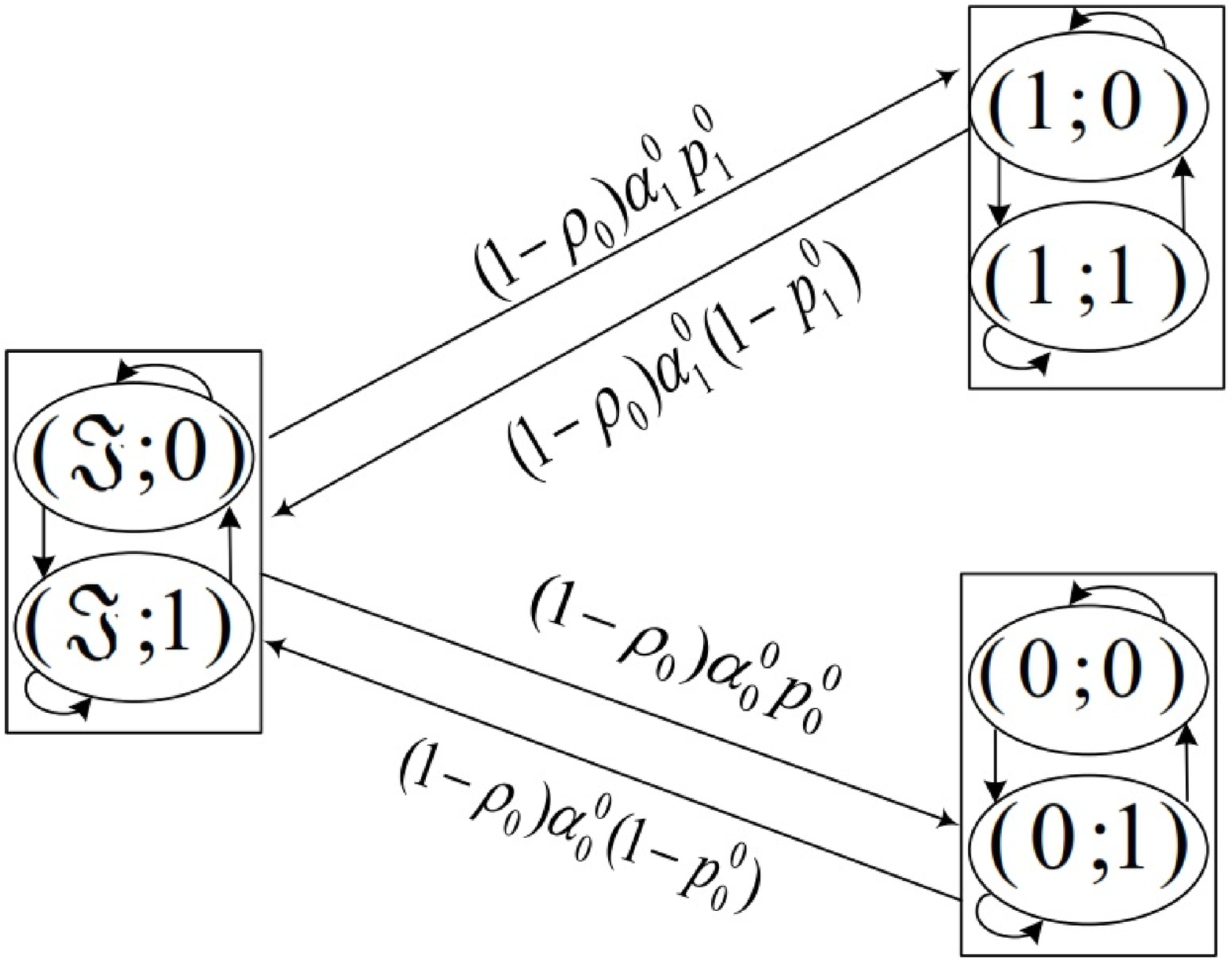}
                \caption{DTMC$^{0}$.}
                \label{markov0}
        \end{subfigure}%
        \hfill
        \begin{subfigure}[t]{0.3\textwidth}
                \includegraphics[width=\textwidth]{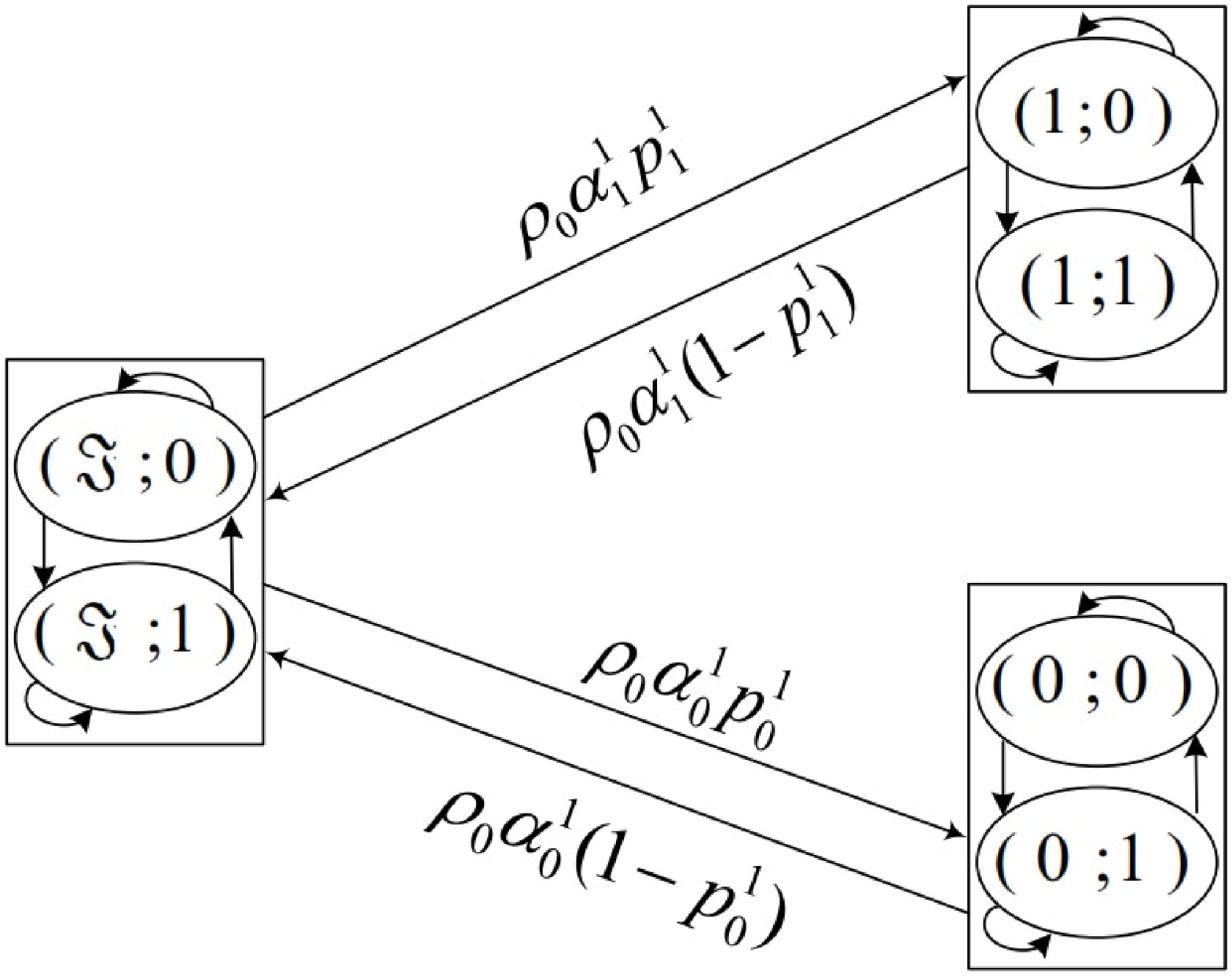}
                \caption{DTMC$^{1}$. }
                \label{markov1}
        \end{subfigure}
				\caption{Two different DTMCs associated with $s_{0}(t)$.}
				\label{markov}
\end{figure}

Inspired by the approach of \cite{qcsmaeylem}, we define a separate DTMC$^{\boldsymbol{\vartheta}}$ for each of the different realizations of channel state vector $\boldsymbol{\vartheta}\in\mathcal{S}(t)$ to mitigate this problem. Let $y^{\boldsymbol{\vartheta}}\in\mathcal{M}^{\boldsymbol{\vartheta}}$ be the states of DTMC$^{\boldsymbol{\vartheta}}$. Note that a transition from state $y^{\boldsymbol{\vartheta}}$ to another state $\hat{y}^{\boldsymbol{\vartheta}}$ in DTMC$^{\boldsymbol{\vartheta}}$ only depends on the current state $y^{\boldsymbol{\vartheta}}$ given the current channel state vector $\textbf{s}(t)$ and the decision schedule. Thus, DTMC$^{\boldsymbol{\vartheta}}$ is Markovian for all $\boldsymbol{\vartheta}\in\mathcal{S}$.

Figure \ref{markov} shows two different DTMCs associated with each of two different states of $s_{0}$ for the network in {\em Example 1}. The transition probabilities follow directly from the definition of RQ-CSMA algorithm. Figure \ref{markov0} represents DTMC$^{0}$ where a ``Rectangle'' corresponds to $y^{0}$ and an oval corresponds to $\left(y^{0};s_{0}\right)$. For example, consider $(1;0)$ in DTMC$^{0}$; being in this state means that the channel of the node $0$ is OFF and only node $1$ is transmitting, while $(1;1)$ means that the channel of the node $0$ is ON and the last time that channel was OFF prior to the current time slot, node $1$ was transmitting. This follows from the definition of $y^{\boldsymbol{\vartheta}}$ in (\ref{yS}). Similarly, Figure \ref{markov1} represents DTMC$^{1}$ where a ``Rectangle'' corresponds to $y^{1}$ and an oval corresponds to $\left(y^{1};s_{0}\right)$. Note that according to RQ-CSMA, any state, $y^{0}$, (i.e., a rectangle) in DTMC$^{0}$ at any time $t$ can make a transition into another state, $\hat{y}^{0}$ ($y^{0}\neq\hat{y}^{0}$), if and only if $s_{0}(t)=0$, and the state freezes otherwise. The same argument follows for DTMC$^{1}$.

Now consider $y^{0}$ and $y^{1}$ as states for DTMC$^{0}$ and DTMC$^{1}$, respectively\footnote{Note that Markov chains with $\left(y^{0};s_{0}\right)$ and $\left(y^{1};s_{0}\right)$ as states are not time reversible. It is easy to verify this observation from Fig.\ref{markov}, where the outgoing probabilities from $\left(y^{0};0\right)$ and $\left(y^{0};1\right)$ to $\left(\hat{y}^{0};0\right)$ are the same whereas the incoming probabilities from $\left(\hat{y}^{0};1\right)$ to $\left(y^{0};s_{0}\right)$ do not exist if $y^{0}\neq\hat{y}^{0}$.}. The stationary distribution of DTMC$^{0}$ is $\pi(\mathfrak{I})= \frac{1}{Z^{0}}$, $\pi(1)= \frac{1}{Z^{0}}\frac{p^{0}_{1}}{1-p^{0}_{1}}$ and $\pi(0)= \frac{1}{Z^{0}}\frac{p^{0}_{0}}{1-p^{0}_{0}}$ where $Z^{0}=1+\frac{p^{0}_{0}}{1-p^{0}_{0}}+\frac{p^{0}_{1}}{1-p^{0}_{1}}$. Similarly the stationary distribution of DTMC$^{1}$ is $\pi(\mathfrak{I})= \frac{1}{Z^{1}}$, $\pi(1)= \frac{1}{Z^{1}}\frac{p^{1}_{1}}{1-p^{1}_{1}}$ and $\pi(0)= \frac{1}{Z^{1}}\frac{p^{1}_{0}}{1-p^{1}_{0}}$ where $Z^{1}=1+\frac{p^{1}_{0}}{1-p^{1}_{0}}+\frac{p^{1}_{1}}{1-p^{1}_{1}}$. Both DTMC$^{0}$ and DTMC$^{1}$ are time reversible and have product form stationary distributions. Moreover, at any given time only one of the DTMCs evolve. Also, the transition probabilities of DTMC$^{0}$ only depend on $p^{0}_{0}$ and $p^{0}_{1}$, whereas the transition probabilities in DTMC$^{1}$ only depends on $p^{1}_{0}$ and $p^{1}_{1}$.  Hence, we can assume that at any time the DTMCs are at steady state.

\vspace{-0.3cm}
\subsection{Proving the Throughput Optimality of RQ-CSMA}
\label{PTO}
In the following, we consider the general network topology with a source and $N$ relay nodes, with all channels to the destination being time varying ON-OFF channels, whereas the channels from node $0$ to all other nodes are always ON.  We show that DTMC$^{\boldsymbol{\vartheta}}$ for all $\boldsymbol{\vartheta}\in\mathcal{S}$ is time reversible with a product-form stationary distribution. Theorem \ref{th1} gives the product-form of the stationary distribution, and finally, Theorem \ref{thop} claims the throughput optimality of RQ-CSMA.

\begin{figure}
	\centering
		\includegraphics[scale=.27]{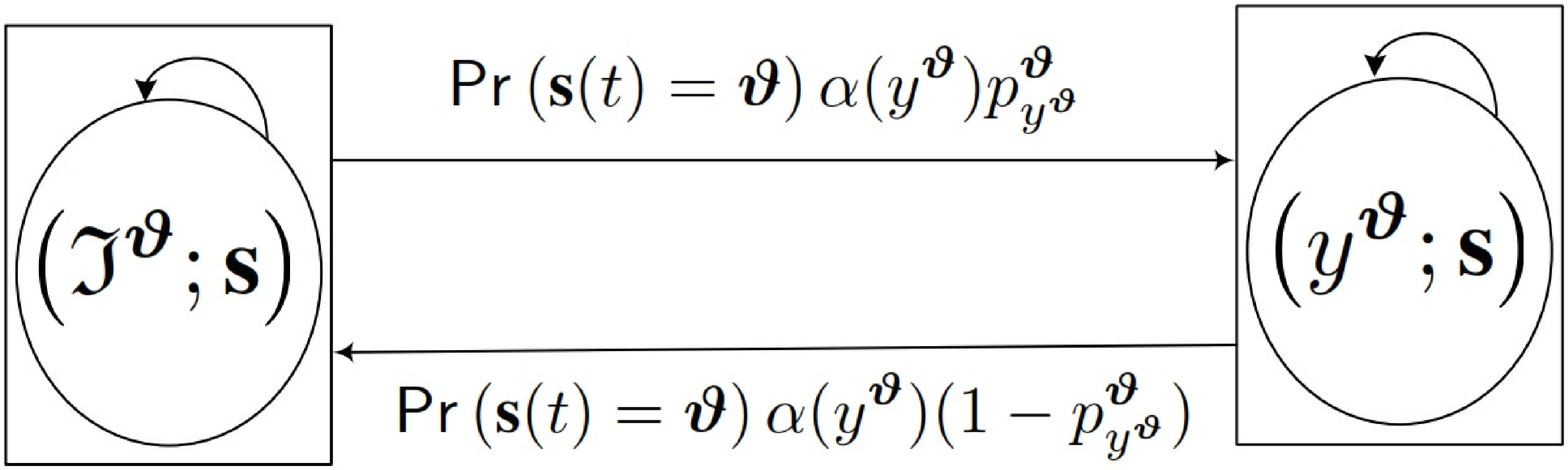}
	\caption{Transition probabilities of states in DTMC$^{\boldsymbol{\vartheta}}$ for any arbitrary $y^{\boldsymbol{\vartheta}}\in \mathcal{M}^{\boldsymbol{\vartheta}}$.}
	\label{markovproof}
\end{figure}

\newtheorem{th1}{Theorem}
\begin{th1}\label{th1}
DTMC$^{\boldsymbol{\vartheta}}$ for any $\boldsymbol{\vartheta}\in\mathcal{S}$ is reversible and it has the following
product-form stationary distribution:
\begin{align}
\pi\left(y^{\boldsymbol{\vartheta}}\right)& = \frac{1}{Z^{\boldsymbol{\vartheta}}}\frac{p^{\boldsymbol{\vartheta}}_{y^{\boldsymbol{\vartheta}}}}{1-p^{\boldsymbol{\vartheta}}_{y^{\boldsymbol{\vartheta}}}},\label{distoff}\nonumber\\
\pi\left(\mathfrak{I}^{\boldsymbol{\vartheta}}\right)& = \frac{1}{Z^{\boldsymbol{\vartheta}}},\\
\mbox{where } Z^{\boldsymbol{\vartheta}}& = 1+\sum_{y^{\boldsymbol{\vartheta}}\in \mathcal{M}^{\boldsymbol{\vartheta}}}\frac{p^{\boldsymbol{\vartheta}}_{y^{\boldsymbol{\vartheta}}}}{1-p^{\boldsymbol{\vartheta}}_{y^{\boldsymbol{\vartheta}}}}.
\end{align}
\end{th1}
\begin{proof}
Consider the state transitions of DTMC$^{\boldsymbol{\vartheta}}$ as given in Fig. \ref{markovproof}. For any $y^{\boldsymbol{\vartheta}}\in\mathcal{M}^{\boldsymbol{\vartheta}}$ we can write the following detailed balance equation which follows from RQ-CSMA:
\begin{align}
&\pi\left(\mathfrak{I}^{\boldsymbol{\vartheta}}\right)\mathsf{Pr}\left(\textbf{s}(t)=\boldsymbol{\vartheta}\right)\alpha(y^{\boldsymbol{\vartheta}})p^{\boldsymbol{\vartheta}}_{y^{\boldsymbol{\vartheta}}}\nonumber\\
&=\pi\left(y^{\boldsymbol{\vartheta}}\right)\mathsf{Pr}\left(\textbf{s}(t)=\boldsymbol{\vartheta}\right)\alpha(y^{\boldsymbol{\vartheta}})(1-p^{\boldsymbol{\vartheta}}_{y^{\boldsymbol{\vartheta}}}),
\end{align}
which simplifies to:
\begin{align}
\pi\left(y^{\boldsymbol{\vartheta}}\right)=\pi\left(\mathfrak{I}^{\boldsymbol{\vartheta}}\right)\frac{p^{\boldsymbol{\vartheta}}_{y^{\boldsymbol{\vartheta}}}}{1-p^{\boldsymbol{\vartheta}}_{y^{\boldsymbol{\vartheta}}}}\label{ydist},
\end{align}
Next we have, 
\begin{align}
&\pi\left(\mathfrak{I}^{\boldsymbol{\vartheta}}\right)+\sum_{y^{\boldsymbol{\vartheta}}\in \mathcal{M}^{\boldsymbol{\vartheta}}}\pi\left(y^{\boldsymbol{\vartheta}}\right)\nonumber\\
&=\pi\left(\mathfrak{I}^{\boldsymbol{\vartheta}}\right)+\sum_{y^{\boldsymbol{\vartheta}}\in \mathcal{M}^{\boldsymbol{\vartheta}}}\pi\left(\mathfrak{I}^{\boldsymbol{\vartheta}}\right)\frac{p^{\boldsymbol{\vartheta}}_{y^{\boldsymbol{\vartheta}}}}{1-p^{\boldsymbol{\vartheta}}_{y^{\boldsymbol{\vartheta}}}}=1.
\end{align}
If we define $Z^{\boldsymbol{\vartheta}} = 1+\sum_{y^{\boldsymbol{\vartheta}}\in \mathcal{M}^{\boldsymbol{\vartheta}}}\frac{p^{\boldsymbol{\vartheta}}_{y^{\boldsymbol{\vartheta}}}}{1-p^{\boldsymbol{\vartheta}}_{y^{\boldsymbol{\vartheta}}}}$, we have:
\begin{align}
\pi\left(\mathfrak{I}^{\boldsymbol{\vartheta}}\right) = \frac{1}{Z^{\boldsymbol{\vartheta}}}\label{Idist}
\end{align}
Substituting (\ref{Idist}) in (\ref{ydist}) gives the stationary distribution of DTMC$^{\boldsymbol{\vartheta}}$.
\end{proof}
The product-form stationary distribution of DTMC$^{\boldsymbol{\vartheta}}$, suggests that we can use the results established in \cite{srikantstable} to prove the throughput optimality of the algorithm. Let $x(t)$ be the schedule at time slot $t$. For the sake of simplicity, we drop $t$ and denote the weight associated with the scheduled node $x(t)$ by $\omega_{x}(t)$. Similarly, $\omega_{y^{\boldsymbol{\vartheta}}}(t)$ is the weight associated with the $y^{\boldsymbol{\vartheta}}(t)$.

\newtheorem{thsrikant}[th1]{Theorem}
\begin{thsrikant}\cite{srikantstable}\label{thsrikant}
Let $\omega^{*}(t):=\max_{x\in\mathcal{M}(t)}\omega_{x}(t)$, where $\mathcal{M}(t)$ is the set of all feasible schedules at time $t$. For a scheduling algorithm, for any $0<\epsilon,\delta<1$, if there exists $\beta>0$ such that if $\omega^{*}(t)>\beta$, the scheduling algorithm chooses a schedule $x(t)\in \mathcal{M}(t)$ that satisfies
\begin{align}
\mathsf{Pr}\left\{\omega_{x}(t)\geq (1-\epsilon)\omega^{*}(t)\right\}\geq 1-\delta
\end{align}
where $\omega_{x}(t)$ is a function of the lengths of queues defined in $(\ref{W0})$ and $(\ref{Wi})$. Then the scheduling algorithm is throughput optimal.
\end{thsrikant}

If we choose $p^{\boldsymbol{\vartheta}}_{y^{\boldsymbol{\vartheta}}}=\frac{\exp(\omega_{y^{\boldsymbol{\vartheta}}}(t))}{1+\exp(\omega_{y^{\boldsymbol{\vartheta}}}(t))}$, then the stationary distribution of DTMC$^{\boldsymbol{\vartheta}}$ becomes:
\begin{align}
&\pi\left(y^{\boldsymbol{\vartheta}}\right) = \frac{\exp(\omega_{y^{\boldsymbol{\vartheta}}}(t))}{Z^{\boldsymbol{\vartheta}}},\nonumber\\
&\pi\left(\mathfrak{I}^{\boldsymbol{\vartheta}}\right) = \frac{1}{Z^{\boldsymbol{\vartheta}}}\\
&Z^{\boldsymbol{\vartheta}} = 1+\sum_{y^{\boldsymbol{\vartheta}}\in \mathcal{M}^{\boldsymbol{\vartheta}}}\exp(\omega_{y^{\boldsymbol{\vartheta}}}(t))
\end{align}

By choosing $f_{i}$'s in $(\ref{W0})$ and $(\ref{Wi})$ wisely, $p^{\boldsymbol{\vartheta}}_{y^{\boldsymbol{\vartheta}}}$ changes slowly over time and we can assume that DTMC$^{\boldsymbol{\vartheta}}$ is in steady-state in every time slot (time scale separation) \cite{qcsma}. In the following we show that RQ-CSMA is throughput optimal by showing that it is close enough to another throughput optimal algorithm (i.e., MWS).
\newtheorem{thop}[th1]{Theorem}
\begin{thop} \label{thop}
For any $\boldsymbol{\vartheta}\in\mathcal{S}$, suppose that $\bigcup_{\mathcal{D}^{\boldsymbol{\vartheta}}\in \mathcal{M}^{\boldsymbol{\vartheta}}}\mathcal{D}^{\boldsymbol{\vartheta}} = \mathcal{M}^{\boldsymbol{\vartheta}}$. Let $p^{\boldsymbol{\vartheta}}_{y^{\boldsymbol{\vartheta}}}=\frac{\exp(\omega_{y^{\boldsymbol{\vartheta}}}(t))}{1+\exp(\omega_{y^{\boldsymbol{\vartheta}}}(t))}$, $\forall y^{\boldsymbol{\vartheta}}\in \mathcal{M}^{\boldsymbol{\vartheta}}$ when $\textbf{s}(t)=\boldsymbol{\vartheta}$. Then RQ-CSMA is throughput-optimal.
\end{thop}
\begin{proof}
We define $\omega^{\boldsymbol{\vartheta}}(t):=\max_{x\in\mathcal{M}^{\boldsymbol{\vartheta}}(t)}\omega_{x}(t)$.
Next define the following states:
\begin{align}
&\chi^{\boldsymbol{\vartheta}} = \left\{\left(y^{\boldsymbol{\vartheta}};\boldsymbol{\vartheta}\right)|s.t.,\ \ \omega_{y^{\boldsymbol{\vartheta}}}(t)<(1-\epsilon)\omega^{\boldsymbol{\vartheta}}(t)\right\},\\
&\Psi^{\boldsymbol{\vartheta}} = \bigcup_{\boldsymbol{\vartheta}\in\mathcal{S}}\chi^{\boldsymbol{\vartheta}}\\
&\varphi^{\boldsymbol{\vartheta}} = \left\{y^{\boldsymbol{\vartheta}}|s.t.,\ \ \omega_{y^{\boldsymbol{\vartheta}}}(t)<(1-\epsilon)\omega^{\boldsymbol{\vartheta}}(t)\right\}.
\end{align}
Note that $\pi\left(\bigcup_{\boldsymbol{\vartheta}\in\mathcal{S}}\left(y^{\boldsymbol{\vartheta}};\boldsymbol{\vartheta}\right)\right)=\pi\left(y^{\boldsymbol{\vartheta}}\right)$ so $\pi(\Psi^{\boldsymbol{\vartheta}})=\pi(\varphi^{\boldsymbol{\vartheta}})$. 
We now calculate the probability of a state in set $\chi^{\boldsymbol{\vartheta}}$:
\begin{align}
\pi(\chi^{\boldsymbol{\vartheta}}) <&\pi(\Psi^{\boldsymbol{\vartheta}})=\pi(\varphi^{\boldsymbol{\vartheta}}) \nonumber\\
= &\sum_{y^{\boldsymbol{\vartheta}}\in\varphi^{\boldsymbol{\vartheta}}}\pi\left(y^{\boldsymbol{\vartheta}}\right)\nonumber\\
=&\sum_{y^{\boldsymbol{\vartheta}}\in\varphi^{\boldsymbol{\vartheta}}}\frac{\exp(\omega_{y^{\boldsymbol{\vartheta}}(t)})}{Z^{\boldsymbol{\vartheta}}}\nonumber\\ \leq &\frac{\left|\chi^{\boldsymbol{\vartheta}}\right|(1-\epsilon)\omega^{\boldsymbol{\vartheta}}(t)}{Z^{\boldsymbol{\vartheta}}}\nonumber\\
<&\frac{{N+2}}{\exp(\epsilon\omega^{\boldsymbol{\vartheta}}(t))},\label{thopeq}
\end{align}
where (\ref{thopeq}) is true because $\left|\chi^{\boldsymbol{\vartheta}}\right|\leq \left|\mathcal{M}^{\boldsymbol{\vartheta}}\right|\leq N+2$, and
\begin{align}
Z^{\boldsymbol{\vartheta}}>\exp(\max_{y^{\boldsymbol{\vartheta}}\in \mathcal{M^{\boldsymbol{\vartheta}}}}\omega_{y^{\boldsymbol{\vartheta}}})>\exp(\omega^{\boldsymbol{\vartheta}}(t)).
\end{align}
Therefore, if
\begin{align}
\omega^{\boldsymbol{\vartheta}}(t)>\frac{1}{\epsilon}\left(\log\left(N+2\right)+\log\frac{1}{\delta^{\boldsymbol{\vartheta}}}\right),
\end{align}
then $\pi(\chi^{\boldsymbol{\vartheta}})<\delta^{\boldsymbol{\vartheta}}$. Then we have the following results:
\begin{align}
&\mathsf{Pr}\left\{\omega_{y^{\boldsymbol{\vartheta}}}(t)\geq (1-\epsilon)\omega^{\boldsymbol{\vartheta}}(t)|\textbf{s}(t)=\boldsymbol{\vartheta}\right\}\nonumber\\
&=1-\mathsf{Pr}\left\{\omega_{y^{\boldsymbol{\vartheta}}}(t)< (1-\epsilon)\omega^{\boldsymbol{\vartheta}}(t)|\textbf{s}(t)=\boldsymbol{\vartheta}\right\}\nonumber\\
&=1-\frac{\frac{N+2}{e^{\epsilon\omega^{\boldsymbol{\vartheta}}(t)}}}{\mathsf{Pr}\left(\textbf{s}(t) =\boldsymbol{\vartheta} \right)}:=1-\delta^{\boldsymbol{\vartheta}}.
\end{align}
Let $x$ be the node scheduled in slot $t$ by RQ-CSMA. Using the total probability law we have:
\begin{align}
&\mathsf{Pr}\left\{\omega_{x}(t)\geq (1-\epsilon)\omega^{*}(t)\right\}\nonumber\\
&=\sum_{\boldsymbol{\vartheta}\in\mathcal{S}}\left\{\right.\mathsf{Pr}\left\{\omega_{y^{\boldsymbol{\vartheta}}}(t)\geq (1-\epsilon)\omega^{\boldsymbol{\vartheta}}(t)|\textbf{s}(t)=\boldsymbol{\vartheta}\right\}\label{eq26}\\
&\phantom{{}=1000}\times \mathsf{Pr}\left(\textbf{s}(t) =\boldsymbol{\vartheta} \right)\left.\right\}\notag\\
&\geq(1-\delta)\sum_{\boldsymbol{\vartheta}\in\mathcal{S}}\mathsf{Pr}\left(\textbf{s}(t) =\boldsymbol{\vartheta} \right)=1-\delta,
\end{align}
where $\delta = \max_{\boldsymbol{\vartheta}\in\mathcal{S}}\delta^{\boldsymbol{\vartheta}}$. Note that $(\ref{eq26})$ holds because $x =y^{\boldsymbol{\vartheta}}$ and $\omega^{*}(t)=\omega^{\boldsymbol{\vartheta}}(t)$ whenever $\textbf{s}(t)=\boldsymbol{\vartheta}$. Hence, Algorithm \ref{alg} satisfies the condition of Theorem \ref{thsrikant}, and thus, it is throughput optimal. 
\end{proof}

\subsection{On the complexity of the RQ-CSMA}
\label{sec:OnTheComplexityOfTheAlgorithms}
At time slot $t$, each node $i$ first needs to to find a decision schedule, and the channel state vector $\textbf{s}(t)$. Note that this requires $N+1+W$ mini slots which results in $O(N)$ complexity. Then in RQ-CSMA each node needs to find the appropriate DTMC$^{\boldsymbol{\vartheta}}$ scheduled for evolution which means that each node needs to find a single element from a set with size $2^{N+1}$. It is well known that the simple binary search, has logarithmic complexity \cite{Ahuja}. Consequently, the total complexity of RQ-CSMA is $O(\log(2^{N+1}))+O(N)=O(N)$.

\section{Achievable Rate Region of the Relay Network}
\label{sec:CapacityRegionOfTheCooperativeNetwork}
In order to validate the throughput optimality of RQ-CSMA, we focus on a network consisting of a source and a relay nodes. This helps us visualize the geometric representation of the achievable rate region and simplify the performance comparison of RQ-CSMA and Q-CSMA. Lemma \ref{capregion} establishes the achievable rate region for two nodes with node $0$ having arrival rate of $\lambda_0$ and probability of ON channel state, $\rho_0$  and node $1$ having arrival rate of $\lambda_1$ and probability of ON channel state, $\rho_1$:

\newtheorem{capregion}{Lemma}
\begin{capregion} \label{capregion}
For two nodes, the achievable rate region of the relay network, when $\rho_1< 0.5$, is
\begin{align}
\Lambda^{c}=\left\{\boldsymbol\lambda|\lambda_{1}<\rho_{1}, \lambda_{0}+\lambda_{1}<\rho_{0}+\rho_{1}(1-\rho_{0})\right\}\label{R1}
\end{align}
and for $\rho_1\geq 0.5$ is
\begin{align}
\Lambda^{c}=\left\{\right.\boldsymbol\lambda|\lambda_{1}<(1-\rho_{0})(2\rho_{1}-1),2\lambda_{0}+\lambda_{1}<1+\rho_{0};\nonumber\\
(1-\rho_{0})(2\rho_{1}-1)\leq\lambda_{1}<\rho_{1},\lambda_{0}+\lambda_{1}<\rho_{0}+(1-\rho_{0})\rho_{1}\left.\right\}\label{R2}
\end{align}
\end{capregion}

\begin{proof}
The proof is given in Appendix.
\end{proof}

The achievable rate region of the relay network, $\Lambda^{c}$, for the case of two nodes with $\rho_1\geq 0.5$ and with $\rho_1< 0.5$ is depicted in Figure \ref{cap3} and \ref{cap2}, respectively. Note that in case of no relay service, node $0$ can only achieve a rate of $\rho_0$, whereas when relaying is enabled, it can achieve an additional rate of $\frac{1-\rho_0}{2}$ for $\rho_2\geq 0.5$ and $\rho_1(1-\rho_0)$ for $\rho_1< 0.5$.

\begin{figure}
	\centering
		\includegraphics[scale=.6]{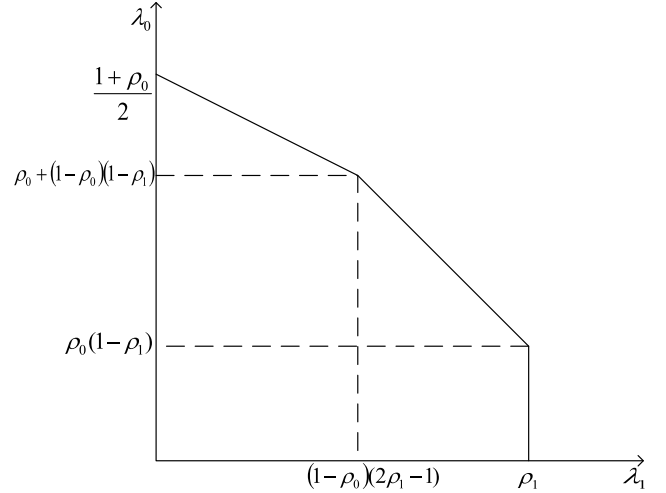}
	\caption{Achievable rate region of the two node relay network with $\rho_1\geq 0.5$. }
	\label{cap3}
\end{figure}

\begin{figure}
	\centering
		\includegraphics[scale=.6]{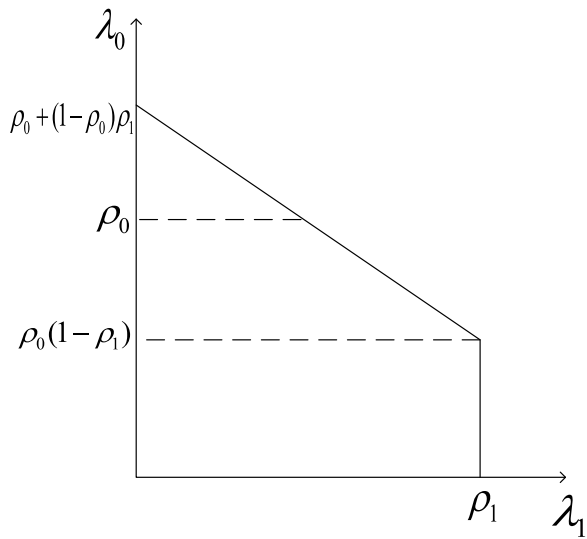}
	\caption{Achievable rate region of the two node relay network with $\rho_1< 0.5$.}
	\label{cap2}
\end{figure}

\section{Numerical Results}
\label{sec:Validnum}

In this section, we numerically validate the throughput optimality of RQ-CSMA and compare its achievable rate region with those of Q-CSMA and a simple algorithm utilizing uniform random back-off \cite{qcsma,qcsmaeylem}. The uniform back-off (UB) algorithm is defined in Algorithm \ref{alg8}. Note that Q-CSMA is not designed for time varying channels so in Q-CSMA if a node is chosen to be activated when its channel is at OFF state, it remains idle in that slot. Also, note that we choose $f_{i}$'s in (\ref{W0}) and (\ref{Wi}) to be of the form $f(x)=\log(\beta x)$ with $\beta=0.1$.

Note that we do not employ the RTD/CTD mechanism in the simulations, since it is required to be implemented by other policies as well to solve the hidden node problem, and its overhead will be similar in all.  Hence, in the following 
 we assume that any pairs of nodes are able to hear each other.


\begin{algorithm}
\caption{Uniform Backoff (UB)}\label{alg8}
\begin{text}
At each time slot,
\end{text}
\begin{algorithmic}[1]
\State Node $i$ selects a random back-off time $T_{i}=\text{Uniform}\left[1,W\right]$ and waits for $T_{i}$ control mini slots.
\State If node $i$ hears an INTENT message from any other node 
before the $(T_{i}+1)$-th control mini-slot, it will not be
included in the transmission schedule and will not
transmit an INTENT message.
\State If node $i$ does not hear an INTENT message from any other node before the $(T_{i}+1)$-th control mini-slot, it will send
an RESV message to all other nodes at the beginning of
$(T_{i}+1)$-th control mini-slot.
\begin{itemize}
	\item If there is a collision, node $i$ will not be included in transmission schedule, $x(t)$.
	\item If there is no collision, node $i$ will be included in the transmission schedule, $x(t)$.
\end{itemize}
\State If $x(t)=i$, then node $i$ will transmit.
(Links with empty queues will keep silent in this time
slot.)
\end{algorithmic}
\end{algorithm}

\subsection{Validating Throughput Optimality by Simulations}
\label{sim22}
We consider a network consisting of two nodes. We evaluate the performance of RQ-CSMA according to the average queue lengths of all queues in the network, i.e., $Q_{avg}=\frac{1}{t}\sum^{t}_{\tau=0}(Q_{0}(\tau)+Q_{1}(\tau)+Q_{01}(\tau))$. To calculate $Q_{avg}$, we run the simulation for $t = 10^{4}$ slots, and over 10 different sample paths. The channel statistics are chosen as $\rho_{0}=0.4$, $\rho_{1}=0.7$. The arrival rates, $(\lambda_1,\lambda_0)$, are taken in the region $\left[0,\rho_1+0.1\right]\times\left[0,\frac{1+\rho_0}{2}+0.1\right]$ where we evaluate $Q_{avg}$ for each point in this region. Figure \ref{MQavgplot} shows the contour graph of $Q_{avg}$ versus $\lambda_0$ and $\lambda_1$, where the solid and thick black line is the achievable rate region of the relay network. We observe that the
network under RQ-CSMA exhibits unstable behavior, shown by the increase in
the average queue lengths (lighter colored contours) as we cross the boundary of the
stability region.

\begin{figure}
	\centering
		\includegraphics[scale=.6]{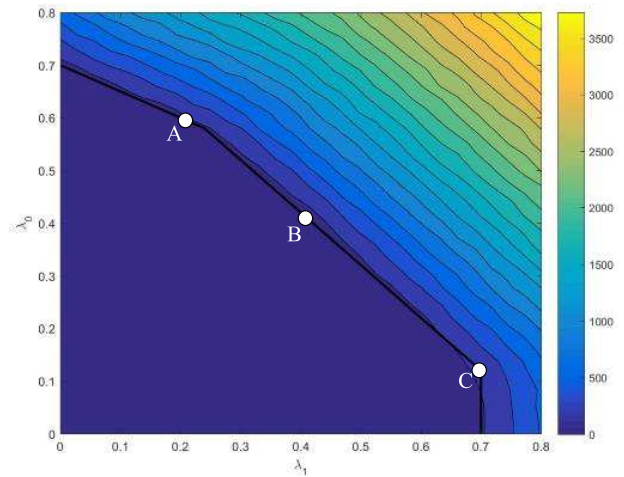}
	\caption{Contour plot of $Q_{avg}$ versus $\lambda_0$ and $\lambda_1$ evaluated by RQ-CSMA.}
	\label{MQavgplot}
\end{figure}

Next, to further investigate the stability region of the RQ-CSMA, we plot the time evolution of the queues at certain points in the achievable rate region. We choose points $A = (0.6,0.2)$, $B=(0.4,0.4)$ and $C=(0.12,0.7)$ on the boundary of the achievable rate region as depicted in Figure \ref{MQavgplot}. Then, we evaluate the time evolution of the queues at $A-\boldsymbol{\epsilon}$, $B-\boldsymbol{\epsilon}$ and $C-\boldsymbol{\epsilon}$, i.e., the points are inside the achievable rate region and they are sufficiently close to the boundary where $\boldsymbol{\epsilon}=0.01$. Also we evaluate the time evolution of the queues for $A+\boldsymbol{\epsilon}$, $B+\boldsymbol{\epsilon}$ and $C+\boldsymbol{\epsilon}$, i.e., the points are outside the achievable rate region and they are sufficiently close to the boundary.
Figures \ref{inner1}, \ref{inner2} and \ref{inner3} show that whenever we take a point inside the achievable rate region close to the boundary, RQ-CSMA is able to stabilize the queues as we expect from the throughput optimality of RQ-CSMA.
\begin{figure}
        \centering
        \begin{subfigure}[t]{0.3\textwidth}
                \includegraphics[width=\textwidth]{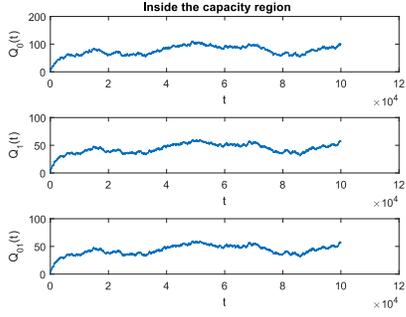}
                \caption{Time evolution of queues at $A-\boldsymbol{\epsilon}=(0.59,0.19)$.}
                \label{inner1}
        \end{subfigure}%
        \hfill
        \begin{subfigure}[t]{0.3\textwidth}
                \includegraphics[width=\textwidth]{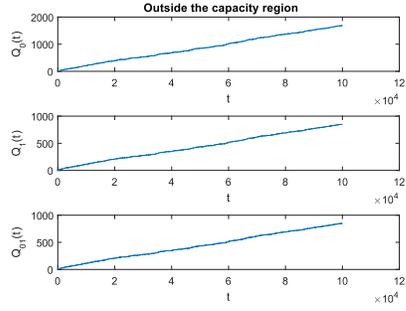}
                \caption{Time evolution of queues at $A+\boldsymbol{\epsilon}=(0.61,0.21)$.}
                \label{outer1}
        \end{subfigure}
				\caption{Time evolution of queues for $A-\boldsymbol{\epsilon}$ and $A+\boldsymbol{\epsilon}$ under RQ-CSMA.}
				\label{MQevol1}
\end{figure}
\begin{figure}
        \centering
        \begin{subfigure}[t]{0.3\textwidth}
                \includegraphics[width=\textwidth]{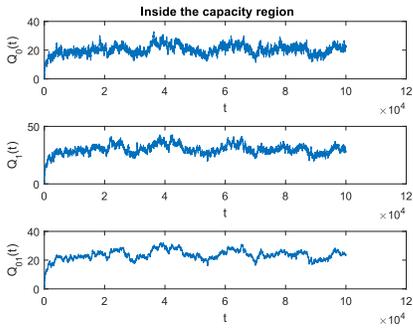}
                \caption{Time evolution of queues at $B-\boldsymbol{\epsilon}=(0.39,0.41)$.}
                \label{inner2}
        \end{subfigure}%
        \hfill
        \begin{subfigure}[t]{0.3\textwidth}
                \includegraphics[width=\textwidth]{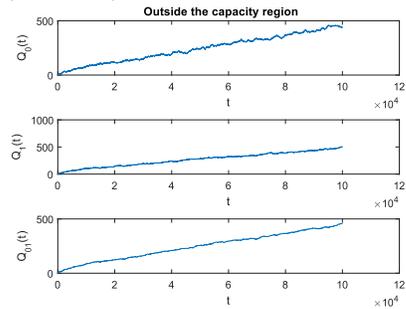}
                \caption{Time evolution of queues at $B+\boldsymbol{\epsilon}=(0.41,0.43)$.}
                \label{outer2}
        \end{subfigure}
				\caption{Time evolution of queues for $B-\boldsymbol{\epsilon}$ and $B+\boldsymbol{\epsilon}$ under RQ-CSMA.}
				\label{MQevol2}
\end{figure}
\begin{figure}
        \centering
        \begin{subfigure}[t]{0.3\textwidth}
                \includegraphics[width=\textwidth]{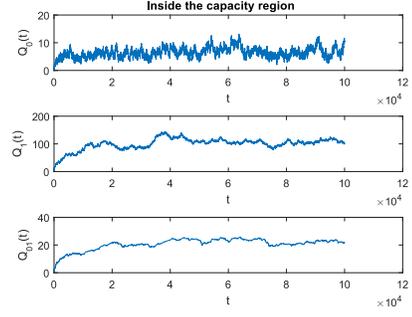}
                \caption{Time evolution of queues at $C-\boldsymbol{\epsilon}=(0.11,0.69)$.}
                \label{inner3}
        \end{subfigure}%
        \hfill
        \begin{subfigure}[t]{0.3\textwidth}
                \includegraphics[width=\textwidth]{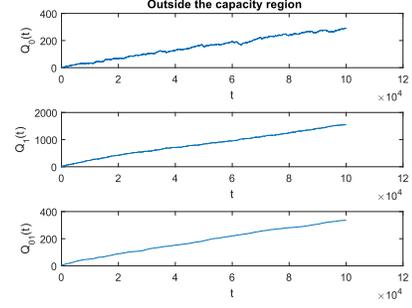}
                \caption{Time evolution of queues at $C+\boldsymbol{\epsilon}=(0.13,0.71)$.}
                \label{outer3}
        \end{subfigure}
				\caption{Time evolution of queues for $C-\boldsymbol{\epsilon}$ and $C+\boldsymbol{\epsilon}$ under RQ-CSMA.}
				\label{MQevol3}
\end{figure}

We repeat the same procedure for Q-CSMA and UB algorithms. Figure \ref{Qavgplot} and \ref{8avgplot} show the contour graphs of $Q_{avg}$ versus $\lambda_0$ and $\lambda_1$ under Q-CSMA and UB, respectively.
We observe that the
network under both Q-CSMA and UB exhibit unstable behavior, shown by the increase in
the average queue lengths, even when the arrivals are far away from the boundaries of the
stability region. To investigate further, we plot the time evolution of the queues under Q-CSMA and UB at $A-\boldsymbol{\epsilon}$. The results are depicted in Figure \ref{Qevol1} and \ref{8evol1}. It can be seen that Q-CSMA and UB are unable to stabilize the queues for arrival rates stabilized by RQ-CSMA.

\begin{figure}
	\centering
		\includegraphics[scale=.45]{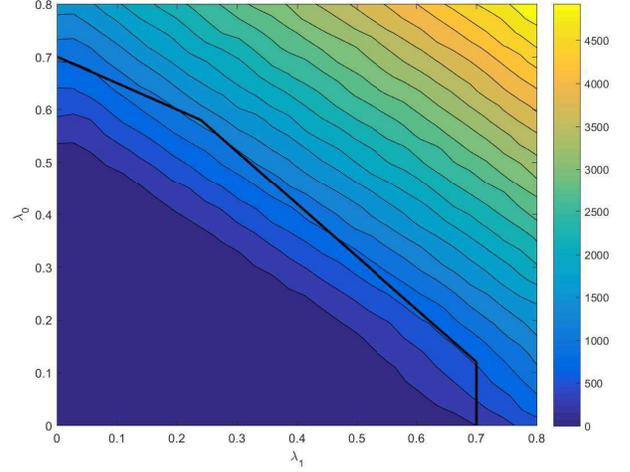}
	\caption{Contour plot of $Q_{avg}$ versus $\lambda_0$ and $\lambda_1$ evaluated by Q-CSMA.}
	\label{Qavgplot}
\end{figure}

\begin{figure}
	\centering
		\includegraphics[scale=.5]{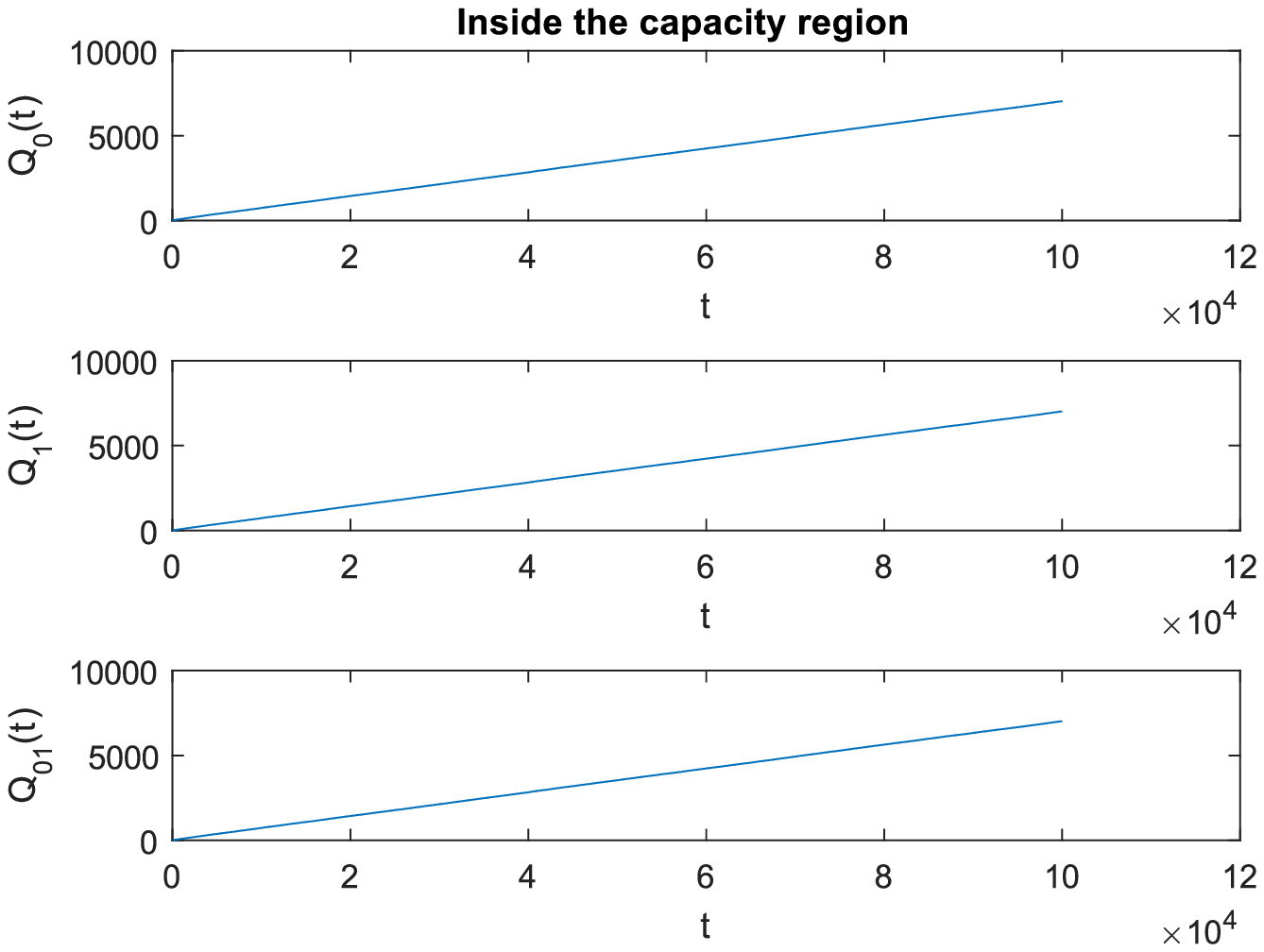}
	\caption{Time evolution of queues for $A-\boldsymbol{\epsilon}=(0.59,0.19)$ under Q-CSMA.}
	\label{Qevol1}
\end{figure}

\begin{figure}
	\centering
		\includegraphics[scale=.45]{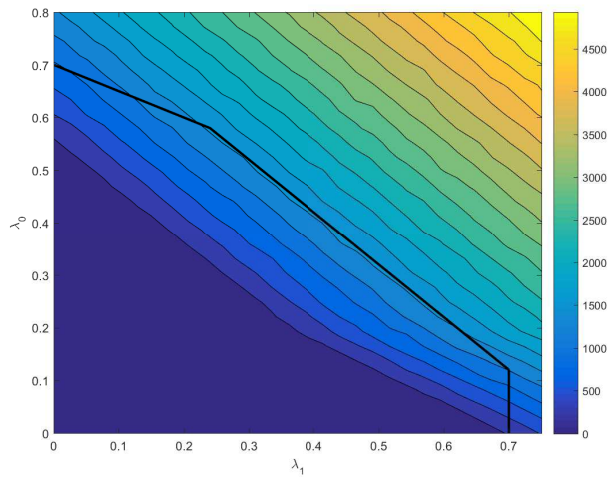}
	\caption{Contour plot of $Q_{avg}$ versus $\lambda_0$ and $\lambda_1$ evaluated by UB.}
	\label{8avgplot}
\end{figure}

\begin{figure}
	\centering
		\includegraphics[scale=.5]{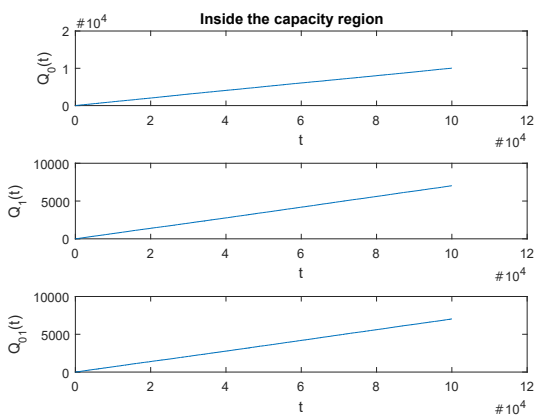}
	\caption{Time evolution of queues for $A-\boldsymbol{\epsilon}=(0.59,0.19)$ under UB.}
	\label{8evol1}
\end{figure}

\subsection{The effect of the number of relay nodes}
\label{sec:ResultsFoMoreThanOneSU}

Next, we numerically compare the performance of RQ-CSMA with UB, MWS and Q-CSMA algorithms when there is more than one relay node. For MWS algorithm, at each time slot, we choose an action for each node which maximizes the weights as defined in (\ref{W0}) and (\ref{Wi}). 
We evaluated the average queue lengths of all nodes for a given arrival rate vector as our performance measure and illustrate $90\%$ confidence intervals.

We consider a network with a source and three relay nodes. We choose $\rho_{0}=0.4$, $\rho_{1}=0.7$, $\rho_{2}=0.8$ and $\rho_{3}=0.7$ for the channel state statistics.

We compare the performance of the algorithms in terms of average queue lengths of all nodes. We begin with an arrival rate vector $\boldsymbol{\lambda} = \left(0.4,\ 0.05,\ 0.05,\ 0.05\right)$ and linearly increase arrival rate of node $0$ from $0.4$ to $0.4+\gamma$. For each algorithm, for a fixed $\gamma$ we run 10 independent simulations and take the average. The average queue lengths of nodes with respect to $\gamma$ is shown in Figure \ref{sim1}, where the vertical bars represent the $90\%$ confidence interval. It can be seen that RQ-CSMA outperforms UB which does not take into account the queue length information and also Q-CSMA which does not take channel state information into account. 

Although the average queue length information gives us an idea about the performance of these algorithms, it does not say much about the stability of the queues. Hence in the following, we compare the algorithms in terms of their queue length evolution with respect to time for a given arrival rate vector.

\begin{figure}
	\centering
		\includegraphics[scale=0.5]{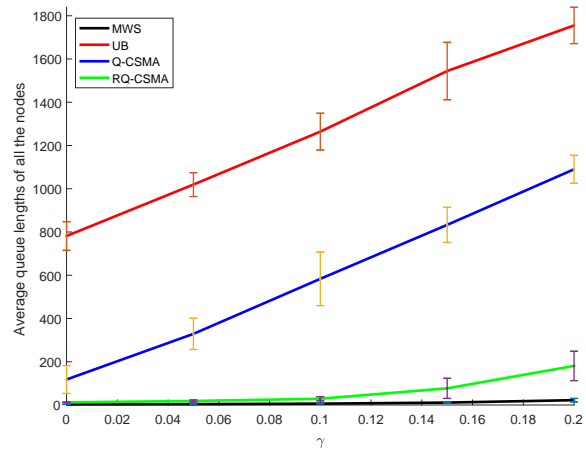}
	\caption{The evolution of the average of the sum of queue lengths vs.\ $\gamma$ over 10 random simulation runs and for an arrival rate vector $\boldsymbol{\lambda} = \left(0.4+\gamma,\ 0.05,\ 0.05,\ 0.05\right)$. The vertical error bars represent the $90\%$ confidence intervals. }
	\label{sim1}
\end{figure}

It can be seen from Figure \ref{sim1} that for $\gamma=0.2$ (i.e., $\lambda_{0} = 0.6$) there is an increase in the average queue lengths of all nodes in all algorithms suggesting that the network is operating at a point close to the boundary of its achievable rate region. We pick the arrival rate vector as $\boldsymbol{\lambda} = \left(0.6,\ 0.05,\ 0.05,\ 0.05\right)$ where the increase in average queue lengths of nodes is apparent in MWS algorithm according to Figure \ref{sim1}, i.e., $\boldsymbol{\lambda}$ is close to the boundary of the stability region. We again run each algorithm for arrival rate vector $\boldsymbol{\lambda}$ for 10 random seeds and then we take their average. The evolution of the queue lengths of nodes (where the sizes of all queues are summed) versus time is given in Figure \ref{sim2}, where the vertical bars represent the $90\%$ confidence interval.  It can be seen that for the desired arrival rate vector, MWS and RQ-CSMA can stabilize the queues whereas in UB and Q-CSMA the queue sizes tend to increase as the time evolves suggesting the instability of the queues.

\begin{figure}
	\centering
		\includegraphics[scale=.5]{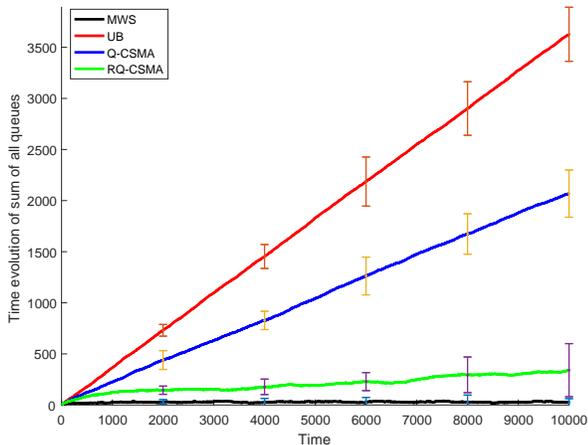}
	\caption{The evolution of the average of the sum of queue lengths vs.\ time over 10 random simulation runs and for an arrival rate vector $\boldsymbol{\lambda} = \left(0.6,\ 0.05,\ 0.05,\ 0.05\right)$. The vertical error bars represent the $90\%$ confidence intervals. }
	\label{sim2}
\end{figure}

\section{Conclusion}
\label{sec:Conclusion}
In this paper, we addressed the problem of scheduling in wireless relay networks with a source and multiple relay nodes where all nodes transmit to a common destination node over independent ON-OFF channels.  In this relay network, $N$ relay nodes help a single node with bad channel quality by relaying its packets. The scheduling is based on a variant of well known Q-CSMA algorithm, and we prove the throughput optimality of the developed algorithm.  We proposed a new contention resolution algorithm, which is used by the nodes to infer which nodes have ON channels; an information required by the scheduling algorithm for achieving throughput optimality. RQ-CSMA is different from Q-CSMA in the sense that it uses channel state information to schedule the nodes and it actually benefits from time variability of the channel states. The performance of the RQ-CSMA, Q-CSMA, MWS and a simple random back-off based algorithms are compared through simulations to demonstrate the efficacy of RQ-CSMA.

As a future work, we will consider the extension of RQ-CSMA to multi-hop and multi source nodes relay networks. Also note that the throughput optimality of RQ-CSMA depends on the assumption that nodes are able to perfectly infer the channel state vector in the contention slot. The investigation of scheduling algorithms with partial and incomplete channel state information is left as future work.

\section*{Proof of Lemma \ref{capregion}}
To simplify the proof, we assume that node $1$ has a single queue and stores the received packets from node $0$ in its own queue. Note that this does not change the stability region of the network. 
Define $\mu_0(t)$ and $\mu_1(t)$ as the instantaneous service rates of node $0$ and node $1$, respectively, and $\mu_{01}(t)$ as the  instantaneous service rate of node $0$ if it relays its packets to node $1$. Note that $\mu_0(t)$, $\mu_1(t)$ and $\mu_{01}(t)$ take binary values $0$ or $1$ and satisfy the following inequality $\mu_0(t)+\mu_1(t)+\mu_{01}(t)\leq 1$. Then, the evolution of the queues are given as:
\begin{align}
Q_0(t+1) = \max\left[Q_0(t)-\mu_0(t)-\mu_{01}(t),0\right]+A_0(t)\nonumber\\
Q_1(t+1) = \max\left[Q_1(t)-\mu_1(t),0\right]+A_1(t)+\mu_{01}(t)
\end{align}

It is well known that when MWS is applied the resulting service rates are throughput optimal. Thus, the expected service rates can be calculated as follow:

\begin{align}
\mathds{E}\left(\mu_0(t)|Q_0(t),Q_1(t)\right) &= \rho_0(1-\rho_1)+\rho_0\rho_1 \mathds{1}_{\left\{Q_0(t)\geq Q_1(t)\right\}}\nonumber\\
\mathds{E}\left(\mu_{01}(t)|Q_0(t),Q_1(t)\right) &= (1-\rho_0)(1-\rho_1)\mathds{1}_{\left\{Q_0(t)-Q_1(t)>0\right\}}\nonumber\\
&+(1-\rho_0)\rho_1\mathds{1}_{\left\{Q_0(t)-Q_1(t)>Q_1(t)\right\}}\nonumber\\
\mathds{E}\left(\mu_1(t)|Q_0(t),Q_1(t)\right) &= (1-\rho_0)\rho_1\mathds{1}_{\left\{Q_0(t)-Q_1(t)\leq Q_1(t)\right\}}\nonumber\\
&+\rho_0\rho_1 \mathds{1}_{\left\{Q_0(t)< Q_1(t)\right\}}
\end{align}

Denote $\textbf{Q}(t) = (Q_0(t),Q_1(t))$. The drift in the Lyapunov function, $L(\textbf{Q}(t)) =\frac{1}{2}(Q^2_0(t)+Q^2_1(t))$, can be written as:
\begin{align}
L(\textbf{Q}(t+1))-L(\textbf{Q}(t)) &= \frac{1}{2}\sum^{1}_{i=0}(Q^2_i(t+1)-Q^2_i(t))\nonumber\\
&\leq \frac{(\mu_0(t)+\mu_{01}(t))^2+A^2_0(t)}{2}\nonumber\\
&+Q_0(t)(A_0(t)-\mu_0(t)-\mu_{01}(t))\nonumber\\
&+\frac{\mu^2_{1}(t)+(A_1(t)+\mu_{01}(t))^2}{2}\nonumber\\
&+Q_1(t)(A_1(t)+\mu_{01}(t)-\mu_1(t))
\end{align}
Now define $B$ as a bound on the summation of first and third terms above as follows:

\begin{align}
\frac{(\mu_0(t)+\mu_{01}(t))^2+A^2_0(t)}{2}+\frac{\mu^2_{1}(t)+(A_1(t)+\mu_{01}(t))^2}{2}\leq B
\end{align}
Now let us assume that $\rho_1<0.5$ so that the achievable rate region is as in (\ref{R1}). First we focus on $\rho_0\leq \lambda_0<\rho_0+(1-\rho_0)\rho_1$ and $0\leq \lambda_1\leq \rho_1(1-\rho_0)$. For small positive value of $\epsilon$ we have $\lambda_0+\lambda_1+\epsilon=\rho_0+(1-\rho_0)\rho_1$. We write the conditional Lyapunov drift as follow:

\begin{align}
&L(\textbf{Q}(t+1)|Q_0(t)<Q_1(t))-L(\textbf{Q}(t)|Q_0(t)<Q_1(t))\leq B\nonumber\\
&+ (\lambda_0-\rho_0(1-\rho_1))Q_0(t)+(\lambda_1-\rho_1)Q_1(t)\nonumber\\
&\leq B+ (\lambda_0-\rho_0(1-\rho_1))Q_0(t)+(\lambda_1-\rho_1)Q_0(t)\nonumber\\
&\leq B-\epsilon Q_0(t) \leq B-\epsilon_1 \sum^{1}_{i=0}Q_{i}(t)
\end{align}
With a similar approach we can show that:
\begin{align}
&L(\textbf{Q}(t+1)|Q_1(t)\leq Q_0(t)\leq 2Q_1(t))\nonumber\\
&-L(\textbf{Q}(t)|Q_1(t)\leq Q_0(t)\leq 2Q_1(t))\nonumber\\
&\leq B-\epsilon_2 \sum^{1}_{i=0}Q_{i}(t)
\end{align}
Expectation over the $\textbf{Q}(t)$ yields:
\begin{align}
\mathds{E}\left(L(\textbf{Q}(t+1))-L(\textbf{Q}(t))\right)<B-\acute{\epsilon}\sum^{1}_{i=0}\mathds{E}(Q_{i}(t))
\end{align}
The above holds for all $t \in\left\{0, 1, 2,\ldots\right\}$ and with the same analysis as above holds for all points inside the achievable rate region defined by (\ref{R1}). Summing over $t \in\left\{0, 1, 2,\ldots\right\}$ for some integer
$T>0$ yields (by telescoping sums):
\begin{align}
\mathds{E}(L(\textbf{Q}(T))-L(\textbf{Q}(0)))\leq BT-\acute{\epsilon}\sum^{T-1}_{t=0}\sum^{1}_{i=0}\mathds{E}(Q_{i}(t))
\end{align}
Rearranging the terms and taking the limit yields:
\begin{align}
\lim_{T\rightarrow\infty}\frac{1}{T}\sum^{T-1}_{t=0}\sum^{1}_{i=0}\mathds{E}(Q_{i}(t))\leq \frac{B}{\acute{\epsilon}}
\end{align}
Thus, all queues are strongly stable. 
The proof of the Lemma for $\rho_1\geq 0.5$ is straightforward and can be done with a similar analysis as above. For more information refer to \cite{Mytez}.

\bibliographystyle{IEEEtran} 
\bibliography{Bibliography} 
\end{document}